\documentclass[12pt]{article}
\usepackage{amsmath,amsfonts,amsthm,amssymb,mathrsfs,graphicx}
\usepackage{times}
\usepackage{color,ulem, cancel, slashed}
\theoremstyle{plain}
\newtheorem{theorem}{Theorem}[section]
\newtheorem{lemma}[theorem]{Lemma}
\newtheorem{lem}[theorem]{Lemma}
\newtheorem{prop}[theorem]{Proposition}
\newtheorem{cor}[theorem]{Corollary}
\theoremstyle{definition}

\newtheorem{defn}[theorem]{Definition}
\theoremstyle{remark}
\newtheorem*{remark}{Remark}
\newtheorem*{rmk}{Remark}

\newcommand{\vphi}{\varphi}

\newcommand{\pl}{\partial}

\newcommand{\na}{\nabla}
\newcommand{\bna}{\bar\nabla}
\newcommand{\lt}{\left}
\newcommand{\rt}{\right}
\newcommand{\rw}{\rightarrow}

\renewcommand{\tilde}{\widetilde}

\newcommand{\tr}{\mbox{tr}}

\newcommand{\RN}[1]{\textup{\uppercase\expandafter{\romannumeral#1}}}

\title{Transformation of mass-angular momentum aspect under BMS transformations}
\author{Po-Ning Chen, Mu-Tao Wang, Ye-Kai Wang, and Shing-Tung Yau}
\date{\today}

\begin{document}

\maketitle
\begin{abstract} 
In this article, we present the definitive transformation formulae of the mass aspect and angular momentum aspect under BMS transformations. Two different approaches that lead to the same formulae are taken. In the first approach, the formulae are derived by reading off the aspect functions from the curvature tensor. While in the second and more traditional approach, we read them off from the metric coefficients. As an application of the angular momentum aspect transformation formula, we directly verify a relation concerning the Dray-Streubel angular momentum. 
It also enables us to reinterpret our calculations in terms of differential forms on null infinity, and leads to an exact expression of the Drey-Streubel angular momentum of a general section.  The formulae we obtained played crucial roles in our recent work on supertranslation invariant charges, and resolved some inconsistencies in the literature.
\end{abstract}

\section{Introduction}
In the pioneering work \cite{BVM, Sachs}, Bondi-van der Burg-Metzner and Sachs introduced a coordinate-based approach, now called the Bondi-Sachs formalism \cite{MW}, to study the gravitational fields at large distance. The Bondi-Sachs coordinate system consists of an optical function $u$ whose level sets are null hypersurfaces, the luminosity radius $r$, and the coordinates of the spherical section $x^A, A=2,3$. In Bondi-Sachs coordinates, the spacetime metric takes the form 
\begin{align}\label{Bondi-Sachs}
-UV du^2 - 2 U dudr + r^2 h_{AB} (dx^A + W^A du)(dx^B+W^B du), A, B=2, 3
\end{align}
where $h_{AB}$ satisfies the {\it determinant condition}
\begin{align*}
\det h_{AB} &= \det \sigma_{AB}.
\end{align*}
Here $\sigma$ denotes the standard metric on $S^2$. The boundary conditions assume that the spacetime is asymptotically Minkowskian and the metric coefficients can be expanded into power series of $\frac{1}{r}$:
\begin{align}\label{boundary condition}
\begin{split}
U &= 1 + \frac{U^{(-2)}}{r^2} + O(r^{-3}),\\ 
V &= 1 - \frac{2m}{r} + O(r^{-2}),\\
W^A &=\frac{W^{A(-2)}}{r^2} + \frac{W^{A(-3)}}{r^3} + O(r^{-4})\\
h_{AB} &= \sigma_{AB} + \frac{1}{r} C_{AB} + \frac{1}{r^2} d_{AB} + O(r^{-3}) \end{split}
\end{align} where the determinant condition implies that $C_{AB}$ is a symmetric traceless 2-tensor on $(S^2,\sigma)$. Each of $U^{(-2)}$, $m$, 
$W^{A(-2)}$, $W^{A(-3)}$, $C_{AB}$, and $d_{AB}$ depends on the coordinates $(u, x^A)$.

We further define\footnote{Our convention is $N^A = -3 N^A_{(CJK)}$. See \cite[(C.104)]{CJK}.}  $N^A$ by way of the following equation: 
\begin{align}\label{defn N^A}
W^{A(-3)} = \frac{2}{3} N^A - C^A_B W^{B(-2)} + \pl_A U^{(-2)}.
\end{align}

\begin{defn}
The function $m(u,x^A)$, 1-form $N_A(u,x^A)=\sigma_{AB} N^B$ and 2-tensor $C_{AB}(u,x^A)$ are called mass aspect, angular momentum aspect and shear tensor. They are regarded as tensors (depending on $u$) on $S^2$. In this note we raise and lower indices with respect to $\sigma_{AB}$.
\end{defn}

\begin{rmk} Our convention of the angular momentum aspect features $N_A = \lim_{r\rw\infty} r^3 R_{Arru}$ where $R_{Arru}=R(\partial_A, \partial_r, \partial_r, \partial_u)$ and $R(X, Y, Z, W) = \langle (D_XD_Y -D_YD_X - D_{[X,Y]})W,Z \rangle$ is the curvature tensor.
\end{rmk}

Working in a coordinate-based formalism, it is important to understand the coordinate transformations that preserve \eqref{Bondi-Sachs} and \eqref{boundary condition} \textemdash the {\it BMS transformation}. The goal of this note is to derive the transformation formula for the mass aspect and the angular momentum aspect under the BMS transformations, Theorem \ref{mass aspect} and Theorem \ref{angular momentum aspect}. In one approach, the mass and angular aspect are regarded as curvature components, and this reduces the amount of computation significantly. Traditionally, the mass and angular momentum aspect are treated as metric components, and one endures the calculation of the sub-leading orders of BMS transformations to obtain the transformation formula. The transformation formula of the angular momentum aspect under supertranslation was first computed in \cite{CJK}. The correct form of this formula is essential in our recent work on supertranslation invariant charges that include angular momentum and center of mass \cite{CWWY1, CWWY2, CKWWY}. Due to the inconsistency of our result with the result in \cite{CJK}, we decide to undertake two different approaches of derivation to make sure the formulae we obtained here are definitive.

The derivation of the formulae is rather complicated and involves quantities such as $U$, $V$, $W^A$, etc. that were originally adopted in \cite{BVM, Sachs} and several more intermediate quantities that were introduced in the calculation (we closely follow the Bondi-Sachs formalism in \cite{MW} though). The following presentation is only in terms of the metric coefficients and their expansions; the correspondence can be found in appendix B.

\begin{theorem}\label{main} Suppose $(\bar{u}, \bar{r}, \bar{x}^A), A=2, 3$ is a Bondi-Sachs coordinate system for a vacuum spacetime such that for $\bar{r}>\bar{r}_0>>1$, the spacetime metric takes the form
\[\bar{g}_{\bar{u}\bar{u}} d\bar{u}^2+2\bar{g}_{\bar{u}\bar{r}} d\bar{u} d\bar{r}+2\bar{g}_{\bar{u}{A}} d\bar{u} d\bar{x}^{ A}+\bar{g}_{{A}{B}} d\bar{x}^{{A}} d\bar{x}^{{B}}\] and that the metric components allow the following expansions: 
\[\begin{split}\bar{g}_{\bar{u}\bar{u}}&=-1+\bar{g}_{\bar{u}\bar{u}}^{(-1)}(\bar{u}, \bar{x}^A)\bar{r}^{-1} +o(\bar{r}^{-1})\\
\bar{g}_{\bar{u}\bar{r}}&=-1+\bar{g}_{\bar{u}\bar{r}}^{(-1)}(\bar{u}, \bar{x}^A)\bar{r}^{-1}+o(\bar{r}^{-1})\\
\bar{g}_{\bar{u}{A}}&=\bar{g}_{\bar{u}{A}}^{(0)}(\bar{u}, \bar{x}^A)+\bar{g}_{\bar{u}{A}}^{(-1)}(\bar{u}, \bar{x}^A)\bar{r}^{-1}+o(\bar{r}^{-1})\\
\bar{g}_{{A}{B}}&=\sigma_{AB}\bar{r}^2+\bar{g}_{{A}{B}}^{(1)}(\bar{u}, \bar{x}^A)\bar{r}+o(\bar{r}^{}).\end{split}\]

Denote by $m(\bar{u}, \bar{x}^A)=\frac{1}{2}\bar{g}_{\bar{u}\bar{u}}^{(-1)}(\bar{u}, \bar{x}^A)$ the mass aspect,  \[ N_A(\bar{u}, \bar{x}^A)=\frac{3}{2}[\bar{g}_{\bar{u}{A}}^{(-1)}+\frac{1}{16}\partial_A (\sigma^{BD}\sigma^{CE} \bar{g}_{BC}^{(1)} \bar{g}_{DE}^{(1)})] (\bar{u}, \bar{x}^A) \] the angular momentum aspect, and $C_{AB}(\bar{u}, \bar{x}^A)=\bar{g}_{{A}{B}}^{(1)}(\bar{u}, \bar{x}^A)$ the shear of the $(\bar{u}, \bar{r}, \bar{x}^A)$ coordinates.

Suppose $({u}, {r}, {x}^a), a=2, 3$ is another Bondi-Sachs coordinate system that is related to $(\bar{u}, \bar{r}, \bar{x}^A)$ by 
\begin{align}
\bar x^A (u,r,x^a) &= g^A(x^a) + \bar{x}^{(-1)A}(u,x^a) r^{-1} + o(r^{-1})\\
\bar u(u,r,x^a) &= K(x^a) \lt( u + f(x^a) \rt) +  \bar{u} ^{(-1)}(u,x^a) r^{-1}+ o(r^{-1}) \label{bar u}\\
\bar r(u,r,x^a )& = \frac{1}{K(x^a)} r + \bar{r}^{(0)}(u,x^a) + o(1),
\end{align} 
where $g^A(x^a)$ is a conformal map on $S^2$ that satisfies:
\begin{align}\label{conformal}
\sigma_{AB} \pl_a g^A \pl_b g^B = K^2(x^a) \sigma_{ab} =: \bar\sigma_{ab}, 
\end{align} such that the spacetime metric in $({u}, {r}, {x}^a)$ near $r=\infty$ takes the form 

\[{g}_{{u}{u}} d{u}^2+2{g}_{{u}{r}}  d{u} d{r}+2{g}_{{u}{a}} d{u} d{x}^{a}+{g}_{ab} d{x}^{a} d{x}^{b}\] and the metric components have the following expansions: 
\[\begin{split}{g}_{{u}{u}}&=-1+{g}_{{u}{u}}^{(-1)}({u}, {x}^a){r}^{-1} +o({r}^{-1})\\
{g}_{{u}{r}}&=-1+{g}_{{u}{r}}^{(-1)}({u}, {x}^a){r}^{-1}+o({r}^{-1})\\
{g}_{{u}{a}}&={g}_{{u}{a}}^{(0)}({u}, {x}^a)+{g}_{{u}{a}}^{(-1)}({u}, {x}^a)r^{-1}+o(r^{-1})\\
{g}_{ab}&={r}^2\sigma_{ab}+{g}_{ab}^{(1)}({u}, {x}^a){r}+o({r}^{})\end{split}\]

Denote by $\mathfrak{m}({u}, {x}^a)=\frac{1}{2}{g}_{{u}{u}}^{(-1)}({u}, {x}^a)$ the mass aspect and \[ \mathfrak{N}_a({u}, {x}^a)=\frac{3}{2}[g_{{u}{a}}^{(-1)}+\frac{1}{16}\partial_a(\sigma^{bd}\sigma^{ce} g_{bc}^{(1)} g_{de}^{(1)})](u, x^a) \] the angular momentum aspect of the $(u, r, x^a)$ coordinate system. 

Denoting $\hat f(u, x^a)= K(x^a) \lt( u + f(x^a) \rt)$, then the following transformation formulae of the mass aspects and angular momentum aspects hold:

\begin{equation}\label{final_mass_formula}
\begin{split}\mathfrak{m}(u,x^a) = K^3 \Big( m + \frac{1}{2}\bna^a \hat f \pl_a g^A \pl_{\bar u} \na^D C_{AD} + \frac{1}{4} \bna^a \hat f \bna^b \hat f \pl_a g^A \pl_b g^B \pl_{\bar u} \pl_{\bar u} C_{AB} \\
+ \frac{1}{4} \bna^a \bna^b \hat f \pl_a g^A \pl_b g^B \pl_{\bar u}C_{AB}  \Big)\end{split}
\end{equation}
where $m, \pl_{\bar u}C_{AB}, \pl_{\bar u }\pl_{\bar u}C_{AB}$ on the right hand side are evaluated at $(\bar u=\hat f(u,x^a), \bar x^A=g^A(x^a))$.

\begin{equation}\label{final_angular_formula}
\begin{split}\mathfrak{N}_a(u,x^a) &= K^2 \Big(\pl_a g^A N_A + 3m \na_a \hat f\\
&\qquad - \frac{3}{4} \pl_a g^A \lt( \na_B \na^D C_{DA} - \na_A \na^D C_{DB} \rt)\pl_c g^{B} \bar\na^c \hat f  - \frac{3}{4} \pl_a g^A \na^B \pl_{\bar u}C_{BA} |\bna\hat f|^2\\
&\qquad +\frac{3}{4} \pl_{\bar u}C^B_D C_{AB} \pl_a g^A   \pl_c g^D \bar\na^c \hat f + \frac{1}{2} \pl_{\bar u} \pl_{\bar u} C_{AB} \pl_c g^{A}\pl_d g^{B} \bar\na^c \hat f \bar\na^d \hat f \na_a \hat f \\
&\qquad  +\frac{3}{2} \pl_{\bar u} \na^D C_{DB} \pl_cg^{B} \bar\na^c \hat f \na_a \hat f -\frac{1}{4} \pl_{\bar u} \pl_{\bar u} C_{AB} \pl_a g^A \pl_c g^{B} \bar\na^c \hat f |\bna \hat f|^2 \Big)\end{split}
\end{equation} where $N_A, m, C_{AB}, \pl_{\bar u}C_{AB}, \pl_{\bar u}\pl_{\bar u}C_{AB}$ on the right hand side are evaluated at $(\bar{u}, \hat f(u,x^a), \bar x^A=g^A(x^a))$.
\end{theorem}

Equation \eqref{final_mass_formula} is proved in Theorem \ref{mass aspect} and equation \eqref{final_angular_formula} is proved in Theorem \ref{angular momentum aspect}.

The note is organized as follows. Section 2 provides the leading orders of Christoffel symbols, curvature tensors, and BMS transformations. They are used to derive the main transformation formulae \eqref{final_mass_formula} and \eqref{final_angular_formula} in Section 3. In Section 4, we present the traditional approach to derive the transformation formulae. In Section 5, we apply the angular momentum aspect transformation formula we obtained to verify a relation concerning the Dray-Streubel angular momentum and identify an exact expression of the Drey-Streubel angular momentum of a general section by a differential 2-form on null infinity. The results are then generalized to the BMS charge of a general section as well. 
There are three appendices: Appendix A presents two identities for symmetric traceless 2-tensors on $S^2$; Appendix B describes the general scheme to compute metric coefficients in the traditional approach; we go over the derivation of transformation formulae of \cite[Appendix C.5]{CJK} and compare with our formulae in Appendix C.

\section{The leading orders of Christoffel symbols, curvature tensors, and BMS transformations}

\subsection{Christoffel symbols and curvature tensors of spacetime}
In this subsection, we assume that the spacetime metric admits power series expansions in $r$ for a given Bondi coordinate and give the leading orders of the Christoffel symbols and the components of curvature tensor. The computation is straightforward. We remind the readers that indices of tensors on $S^2$ are raised and lowered with respect to the standard metric $\sigma_{AB}$. 
Moreover, we denote the covariant derivative of $\sigma_{AB}$ by $\na$.

\begin{lem}
The Christoffel symbols $\mathbf{\Gamma}_{\alpha\beta}^\gamma$ of the spacetime metric are given by 
\begin{align*}
\mathbf\Gamma^u_{uu} &= r^{-2} \lt( \pl_u U^{(-2)} - m \rt)   + O(r^{-3}) \\
\mathbf\Gamma^u_{ru} &= \mathbf\Gamma^u_{rr} =\mathbf\Gamma^u_{rA}=0\\
\mathbf\Gamma^u_{uA} &= r^{-2} \cdot \frac{1}{2} \lt( \pl_A U^{(-2)} -  W^{(-3)}_A - C_{AB} W^{(-2)B} \rt) + O(r^{-3}) \\
\mathbf\Gamma^u_{AB} &= r\sigma_{AB} + \frac{1}{2}C_{AB} + O(r^{-1})
\end{align*}
For $\Gamma^r_{\alpha \beta}$, we have
\begin{align*}
\mathbf\Gamma^r_{uu} &= r^{-1}\cdot (-\pl_u m)   + O(r^{-2}) \\
\mathbf\Gamma^r_{ur} &= r^{-2}\cdot m   + O(r^{-3})\\
\mathbf\Gamma^r_{uA} &= r^{-1}\lt( -\pl_A m + \frac{1}{2}W^{B(-2)}\pl_u C_{AB} \rt)  + O(r^{-2})\\
\mathbf\Gamma^r_{rr} &= r^{-3}\cdot ( - 2U^{(-2)})  + O(r^{-4})\\
\mathbf\Gamma^r_{rA} &= r^{-1} W^{(-2)}_A   + O(r^{-2})\\
\mathbf\Gamma^r_{AB} &= r\lt( \frac{1}{2}\pl_u C_{AB} - \sigma_{AB}\rt)  + O(1)
\end{align*} 
For $\Gamma^A_{\alpha \beta}$, we have
\begin{align*}
\mathbf\Gamma^A_{uu} &= r^{-2} \pl_u W^{(-2)A} + O(r^{-3})\\
\mathbf\Gamma^A_{ur} &= r^{-4} \cdot \frac{1}{2} \lt( \na^A U^{(-2)} - W^{(-3)A} - C^A_B W^{(-2)B} \rt)  + O(r^{-5})\\
\mathbf\Gamma^A_{uB} &= \frac{1}{2r}\pl_u C^A_B + \frac{1}{2r^2} \lt( \pl_u d^A_B -C^{AD} \pl_u C_{BD} +  \na_B W^{(-2)A} - \na^A W^{(-2)}_B \rt)  + O(r^{-3})\\
\mathbf\Gamma^A_{rr} &=0  \\
\mathbf\Gamma^A_{rB} &= r^{-1}\delta^A_B  + O(r^{-2})\\
\mathbf\Gamma_{AB}^C &= \Gamma_{AB}^C + O(r^{-1})\end{align*}
where $\Gamma_{AB}^C = \frac{1}{2} \sigma^{AD}(\pl_B \sigma_{CD} + \pl_C \sigma_{BD} - \pl_D \sigma_{BC})$ is the Christoffel symbol of $\sigma_{AB}$ on $S^2$.
\end{lem}
 
\begin{prop}\label{curvature}
The curvature tensors of the spacetime metric $R_{\alpha\beta\;\;\delta}^{\;\;\;\;\gamma} = \pl_\alpha \mathbf\Gamma^\gamma_{\beta\delta} - \pl_\beta \mathbf\Gamma^\gamma_{\alpha\delta} + \mathbf\Gamma^\gamma_{\alpha\sigma} \mathbf\Gamma^\sigma_{\beta\delta} - \mathbf\Gamma^\sigma_{\alpha\delta} \mathbf\Gamma^\gamma_{\beta\sigma}$ are given by
\begin{align*}
R_{uAuB} &= -\frac{r}{2} \pl_u\pl_u C_{AB} + O(1)\\
R_{Auru} &= -\frac{1}{r} \pl_u W^{(-2)}_A  + O(r^{-2}) \\
R_{urru} &= \frac{2}{r^3} \lt( m - \pl_u U^{(-2)} \rt)   + O(r^{-4})\\
R_{ABCD} &= O(r) \\
R_{ABru} &= \frac{1}{r} \lt( \frac{1}{4}C_A^D \pl_u C_{DB} -\frac{1}{4} C_B^D \pl_u C_{DA} + \na_A W^{(-2)}_B - \na_B W^{(-2)}_A \rt)   + O(r^{-2})\\
R_{ArBu} &= \frac{1}{r} \lt( m\sigma_{AB} -\frac{1}{4} C_B^D \pl_u C_{DA} + \frac{1}{2} \pl_u d_{AB} + \frac{1}{2} \na_A W^{(-2)}_B - \frac{1}{2} \na_B W^{(-2)}_A \rt)  + O(r^{-2})\\
R_{Arru} &= \frac{1}{r^3} \cdot \frac{3}{2} \lt( C_{AB} W^{(-2)B} + W^{(-3)}_A	-\pl_A U^{(-2)} \rt) = \frac{1}{r^3} N_A + O(r^{-4})\\
R_{ABDu} &= \frac{r}{2} \lt( \na_A \pl_u C_{BD} - \na_B \pl_u C_{AD} \rt) + O(1)\\
R_{rArB} &= O(r^{-3})  
\end{align*}
\end{prop}

In particular, we identified the mass aspect $m$ in $R_{urur}$ and the angular momentum aspect $N_A$ in $R_{Arru}$. 

\subsection{BMS transformations}
 
Let $\vphi: (u,r,x^a) \rightarrow (\bar u, \bar r, \bar x^A)$ be a BMS transformation from one Bondi-Sachs coordinate system $(u,r,x^a)$ to another  $(\bar u, \bar r, \bar x^A)$. It is well-known that the leading order is given by
\begin{align}
\bar x^A (u,r,x^a) &= g^A(x^a) + \frac{1}{r} g^{(-1)A}(u,x^a) + O(r^{-2})\\
\bar u(u,r,x^a) &= K(x^a) \lt( u + f(x^a) \rt) + \frac{1}{r} a^{(-1)}(u,x^a) + O(r^{-2})\\
\bar r(u,r,x^a )& = \frac{1}{K(x^a)} r + \rho^{(0)}(u,x^a) + O(r^{-1}),
\end{align} 
where $g^A(x^a)$ is a conformal map on $S^2$ that satisfies:
\begin{align}\label{conformal}
\sigma_{AB} \pl_a g^A \pl_b g^B = K^2(x) \sigma_{ab} =: \bar\sigma_{ab}.
\end{align}
See Sachs \cite[Section 3(b)]{Sachs}. The notation here is adopted from \cite[(28)]{BVM}.

\begin{remark}
When $f$ is a linear combination of constant and first eigenfunctions on $S^2$, the associated BMS transformation is a translation. See the Appendix in \cite{2023MTWang}. 
\end{remark}

\begin{defn}
We use the shorthand $\hat f = K (u + f)$. We denote the covariant derivative and Laplacian with respect to $\bar\sigma$ in \eqref{conformal} by $\bna$ and $\bar\Delta$ and write $|\bna \hat f|^2 = \bar\sigma^{ab} \pl_a \hat f \pl_b \hat f$. Moreover, we write $C_{AB}$ (instead of $\bar C_{AB}$) for the shear tensor of $(\bar u, \bar r, \bar x^A)$ coordinate system.  
\end{defn}

By the chain rule, we obtain:
\begin{equation}\label{tangent_vec_ext_1}
\begin{split}
\vphi_*(\pl_u) &= K \pl_{\bar u} + \pl_u \rho^{(0)} \pl_{\bar r} + \frac{1}{r} \pl_u g^{(-1)A} \pl_A + \mbox{lower order term}\\
\vphi_*(\pl_r) &= - \frac{1}{r^2} a^{(-1)} \pl_{\bar u} + K^{-1} \pl_{\bar r} - \frac{1}{r^2} g^{(-1)A} \pl_A + \mbox{lower order term}\\
\vphi_*(\pl_a) &= \na_a \hat f \pl_{\bar u} - rK^{-2} \na_a K \pl_{\bar r} + \pl_a g^A \pl_A + \mbox{lower order term}.
\end{split}
\end{equation}

Writing $g_{\alpha\beta}$ for the metric components in $(u,r,x^a)$ coordinate system, we have
\begin{align*}
g_{ra} &= - K^{-1} \na_a \hat f - K^{-2} \sigma_{AB} \pl_a g^A g^{(-1)B} + O(r^{-1})\\
g_{rr} &= \frac{1}{r^2} \lt( 2K^{-1} a^{(-1)} + K^{-2} \sigma_{AB} g^{(-1)A} g^{(-1)B} \rt) + O(r^{-3})\\
g_{ab} &= r^2 \sigma_{ab} \\
&\quad + r \Big[ K^{-2} \na_a K \na_b \hat f + K^{-2} \na_b K \na_a \hat f + 2K\rho^{(0)}\sigma_{ab} \\
&\quad\qquad +  K^{-2} \sigma_{AB} \lt(\pl_a g^A \pl_b g^{(-1)B} + \pl_a g^{(-1)A} \pl_b g^B \rt) \\
&\quad\qquad+  \lt( K^{-2} \pl_D \sigma_{AB} g^{(-1)D} + K^{-1} C_{AB} \rt) \pl_a g^A \pl_b g^B \Big] + O(1).
\end{align*}

We solve inductively 
\begin{enumerate}
\item $g^{(-1)A}$ from $g_{ra}^{(0)}=0$,
\item $a^{(-1)}$ from $g_{rr}^{(-2)}=0$,
\item $\rho^{(0)}$ from $\sigma^{ab} g^{(1)}_{ab}=0$.
\end{enumerate}

First of all, one can verify the following formula by \eqref{conformal}.
\begin{lemma}\label{g(-1)}
\begin{align}\label{g^-1}
g^{(-1)A} = -K \bar\na^c \hat f \pl_c g^A
\end{align}
\end{lemma}

Next, from $g_{rr}^{(0)}=0$ and the above formula, we obtain 
\begin{lemma}
\begin{align*}
a^{(-1)} = - \frac{1}{2} K |\bar\na \hat f|^2
\end{align*}
\end{lemma}
Finally, we get
\begin{lem}\label{rho(0)}
\begin{align*}
\rho^{(0)} &= \frac{1}{2} \bar\Delta \hat f  
\end{align*}
\end{lem}
\begin{proof}
We compute 
\begin{align*}
\pl_b g^{(-1)B} &= - \na_b K \bna^c \hat f \pl_c g^B - K \pl_b(\bna^c \hat f) \pl_c g^B - K \bna^d \hat f \pl_b \pl_d g^B.
\end{align*}
By the transformation of Christoffel symbols (here $\bar\Gamma_{bd}^c$ stands for the Christoffel symbol of $\bar \sigma$ on $S^2$) 
\begin{equation} \bar\Gamma_{bd}^c = \pl_E g^c \lt( \Gamma_{AD}^E \pl_b g^A \pl_d g^D + \pl_b\pl_d g^E \rt), \end{equation}
we get
\begin{align*}
\pl_b\pl_d g^B = \bar\Gamma_{bd}^c \pl_c g^B - \Gamma_{ED}^B \pl_b g^E \pl_d g^D
\end{align*}
and
\begin{align*}
\pl_b g^{(-1)B} &= - \na_b K \bar\na^c \hat f \pl_c g^B - K \bar\na_b \bar\na^c \hat f \pl_c g^B + K \bna^d \hat f \Gamma_{ED}^B \pl_b g^E \pl_d g^D\\
&= K^{-1} \na_b K g^{(-1)B} - K \bna_b \bna^c \hat f \pl_c g^B - \Gamma_{ED}^B \pl_b g^E g^{(-1)D}.
\end{align*}
Since $\sigma_{AB} \pl_a g^A = \bar\sigma_{ac} \pl_B g^c$, we get
\begin{align*}
&K^{-2}\sigma_{AB} \lt( \pl_a g^A \pl_b g^{(-1)B} + \pl_a g^{(-1)A} \pl_b g^B \rt) \\
&= -2 K^{-1} \bna_a \bna_b \hat f - K^{-2} \na^b K \na_a \hat f - K^{-2} \na_a K \na_b \hat f  \\
&\quad + K^{-1} \bna^d \hat f \sigma_{AB} \lt( \pl_a g^A \Gamma_{ED}^B\ \pl_b g^E + \pl_b g^B \Gamma_{ED}^A \pl_a g^E \rt) \pl_d g^D 
\end{align*}
Since $\pl_C \sigma_{AB} = \sigma_{AD} \Gamma_{BC}^D + \sigma_{BD} \Gamma_{AC}^D$, we get 
\begin{align*}
g^{(1)}_{ab} &= 2K^{-1} \rho^{(0)} \bar\sigma_{ab} - 2K^{-1} \bna_a \bna_b \hat f + K^{-1} C_{AB} \pl_a g^A \pl_b g^B
\end{align*} 
and solve $\rho^{(0)}$ from $\sigma^{ab} g^{(1)}_{ab}=0$. 
\end{proof}

As a consequence, we obtain the transformation of the shear and news tensors.
\begin{cor}
The shear tensor transforms as 
\begin{align}\label{shear}
g^{(1)}_{ab}(u,x) = -K^{-1} \lt( 2 \bna_a \bna_b \hat f - \bar\Delta \hat f \bar\sigma_{ab} \rt) + K^{-1} C_{AB}(\hat f(u,x), g(x)) \pl_a g^A \pl_b g^B
\end{align} 
The news tensor transforms as \begin{align}\label{news}
\pl_u g^{(1)}_{ab} = \pl_{\bar{u}} C_{AB}(\hat f(u,x), g(x)) \pl_a g^A \pl_b g^B.
\end{align} 
\end{cor} 

\subsection{Einstein vacuum equation and the outgoing radiation condition}
In the remaining portion of this article, we assume the spacetime satisfies the vacuum Einstein equations $Ric =0.$ See M\"{a}dler-Winicour \cite[Section 3]{MW} for a detailed summary of the Einstein equation in the Bond-Sachs formalism. In this subsection, we recall some of the formulae we will use in the remaining part of this article.

First, in a Bondi-Sachs coordinate, the vacuum Einstein equations imply
\begin{align}
U^{(-2)} &= - \frac{1}{16} C_{DE}C^{DE} \label{U-2}\\
W^{A(-2)} &= \frac{1}{2} \na_B C^{AB} \label{W-2}.
\end{align}
In the previous subsections, we assumed that the spacetime metric admitting an power series expansion in $r$. It is well-known that in Bondi-Sachs formalism, logarithmic terms typically appear in the expansion. 
To guarantee that there is no logarithmic term in the expansion so that we indeed have power series expansion, we assume the ``outgoing radiation condition.''

Combing with the Einstein equation, it implies
\begin{align}\label{dAB}
d_{AB} = \frac{1}{2} C_A^D C_{DB}
\end{align}

\section{Transformation of mass aspect and angular momentum aspect}
Suppose the metric tensors in $(u,r,x^a)$ and $(\bar u,\bar r, \bar x^A)$ coordinate system are given by
\begin{equation}\begin{split}\label{Bondi-Sachs2}
&- \mathcal{U} \mathcal{V} du^2 -2 \mathcal{U} dudr + g_{ab}\lt( dx^a + \mathcal{W}^a du \rt)\lt( dx^b + \mathcal{W}^b du \rt)\\
&\text{and}\\
&-UV d\bar u^2 - 2U d\bar u d\bar r + \bar r^2 h_{AB} \lt( d\bar x^A + W^A d \bar u\rt)\lt( d\bar x^B + W^B d \bar u \rt)
\end{split}\end{equation} 
respectively.

We first derive the transformation formula of the mass aspect. 
\begin{theorem}\label{mass aspect}
Transformation of the mass aspect is given by
\begin{equation}
\begin{split} \label{transformation_mass_general}
\mathfrak{m}(u,x) = K^3 \Big( m + \bna^a \hat f \pl_a g^A \pl_{\bar u} W^{(-2)}_A + \frac{1}{4} \bna^a \hat f \bna^b \hat f \pl_a g^A \pl_b g^B \pl_{\bar u} \pl_{\bar u} C_{AB} \\
+ \frac{1}{4} \bna^a \bna^b \hat f \pl_a g^A \pl_b g^B \pl_{\bar u}C_{AB}  \Big)
\end{split}
\end{equation}
where $m, \pl_{\bar u} W_A^{(-2)}, \pl_{\bar u}C_{AB}, \pl_{\bar u }\pl_{\bar u}C_{AB}$ are evaluated at $(\hat f(u,x), g(x))$.
\end{theorem}
Recalling the definition of $W_A^{(-2)}$ from \eqref{W-2}, we obtain \eqref{final_mass_formula}. 
\begin{proof}
Let $\vphi$ be a BMS transformation. The tensorial nature of the curvature tensor implies
\begin{align*}
R_{urru} &= \tilde{R} (\vphi_*(\pl_u), \vphi_*(\pl_r), \vphi_*(\pl_r), \vphi_*(\pl_u)).
\end{align*}
Here $\bar R$ denotes the curvature tensor in $(\bar u, \bar r, \bar x^A)$ coordinate system. From \eqref{tangent_vec_ext_1} and Proposition \ref{curvature}, we obtain
\begin{align*}
R_{urru} = \bar R_{\bar u \bar r \bar r \bar u} - K \frac{g^{(-1)A}}{r^2} \bar R_{\bar u \bar r A \bar u} - K \frac{g^{(-1)A}}{r^2} \bar R_{\bar u A \bar r \bar u} + K^2 \frac{1}{r^4} g^{(-1)A} g^{(-1)B} \bar R_{\bar u AB \bar u} + O(r^{-4}).
\end{align*}
By Proposition \ref{curvature}, we have
\begin{align*}
\frac{1}{r^3} \lt( 2 \mathfrak{m} - 2 \pl_u \mathcal{U}^{(-2)} \rt) + O(r^{-4}) &= \frac{1}{\bar r^3} \lt( 2m - 2 \pl_{\bar u} U^{(-2)} \rt) \\
&\quad - \frac{2K}{r^2} g^{(-1)A} \cdot \frac{1}{\bar r} \pl_{\bar u} W^{(-2)}_A \\
&\quad + \frac{K^2}{r^4} g^{(-1)A}g^{(-1)B} \cdot \frac{\bar r}{2} \pl_{\bar u} \pl_{\bar u} C_{AB} + O(r^{-4}). 
\end{align*}
We obtain
\begin{align*}
\mathfrak{m} + \frac{1}{16} \pl_u \lt( \sigma^{ac} \sigma^{bd} g^{(1)}_{ab} g^{(1)}_{cd} \rt) &= K^3 \lt( m + \frac{1}{16} \pl_{\bar u} (C_{DE}C^{DE}) \rt) + K^3 \bna^a \hat f \pl_a g^A \pl_{\bar u} W^{(-2)}_A  \\
&\quad + \frac{1}{4 }K^3 \bna^a \hat f \bna^b \hat f \pl_a g^A \pl_b g^B \pl_{\bar u} \pl_{\bar u} C_{AB}
\end{align*}
The assertion follows by noting the following equality which can be derived from \eqref{shear}.
\begin{align*}
\pl_u \lt( \sigma^{ac} \sigma^{bd} g^{(1)}_{ab} g^{(1)}_{cd} \rt) = K^3 \pl_{\bar u} (C_{DE}C^{DE}) - 4 K^3 \pl_{\bar u}C_{AB} \pl_a g^A \pl_b g^B \bna^a \bna^b \hat f 
\end{align*}
\end{proof}

The above formula gets simplified if written in terms of the modified mass aspect function $m - \frac{1}{4}\na^A\na^B C_{AB}$ discussed in \cite{W}. 
\begin{theorem}\label{modified}

\begin{align}\label{modified2}
\mathfrak{m} - \frac{1}{4} \na^a\na^b g^{(1)}_{ab} = K^3 \lt( m - \frac{1}{4} \na^A \na^B C_{AB} + \frac{1}{4} \bar\Delta(\bar\Delta + 2)\hat f \rt), 
\end{align} where $m$ and  $\na^A \na^B C_{AB}$ on the right hand side are evaluated at $(\bar{u}=\hat f(u, x^a), \bar{x}^A=g^A(x^a))$. 
\end{theorem}
\begin{proof}
It is straightforward to check the double divergence of symmetric traceless tensor transforms as
\begin{align*}
\na^a \na^b g^{(1)}_{ab} &= K^3 \bna^a \bna^b (K g^{(1)}_{ab}).
\end{align*}
By Corollary \ref{shear}, we have
\begin{align*}
\bna^a \bna^b (K g^{(1)}_{ab}) &= \bna^a \lt( \pl_{\bar u}C_{AB} \bna^b \hat f \pl_a g^A \pl_b g^B + \na^B C_{AB} \pl_a g^A - \bna_a (\bar\Delta + 2)\hat f \rt)\\
&= \pl_{\bar u}\pl_{\bar u}C_{AB} \bna^a \hat f \bna^b \hat f \pl_a g^A \pl_b g^B + \pl_{\bar u} \na^A C_{AB} \bna^b \hat f \pl_b g^B \\
&\quad + \pl_{\bar u} C_{AB} \bna^a \bna^b \hat f \pl_a g^A \pl_b g^B\\
&\quad +\pl_{\bar u} \na^B C_{AB} \bna^a \hat f \pl_a g^A + \na^A \na^B C_{AB} - \bar\Delta (\bar \Delta + 2)\hat f. 
\end{align*}
By Theorem \ref{mass aspect} and $W_A^{(-2)} = \frac{1}{2} \na^B C_{AB}$, the assertion follows.
\end{proof}

Next we derive the transformation formula of the angular momentum aspect.
\begin{theorem}\label{angular momentum aspect}
Transformation of the angular momentum aspect is given by
\begin{equation}
\begin{split} \label{transformation_angular_general}
\mathfrak{N}_a(u,x) &= K^2 \pl_a g^A N_A + 3mK^2 \na_a \hat f\\
&\quad + \frac{3}{2}K \pl_a g^A \lt( \na_B W^{(-2)}_A - \na_A W^{(-2)}_B \rt)g^{(-1)B} - \frac{3}{2}K^2 \pl_a g^A \pl_{\bar u} W^{(-2)}_A |\bna\hat f|^2\\
&\quad - \frac{3}{4}K \pl_{\bar u}C^B_D C_{AB} \pl_a g^A g^{(-1)D} + \frac{1}{2} \pl_{\bar u} \pl_{\bar u} C_{AB} g^{(-1)A}g^{(-1)B} \na_a \hat f \\
&\quad  - 3K \pl_{\bar u} W^{(-2)}_A g^{(-1)A} \na_a \hat f + \frac{1}{4} K \pl_{\bar u} \pl_{\bar u} C_{AB} \pl_a g^A g^{(-1)B} |\bna \hat f|^2
\end{split}
\end{equation}
where $N_A, m, C_{AB}, \pl_{\bar u} W_A^{(-2)}, \pl_{\bar u}C_{AB}, \pl_{\bar u}\pl_{\bar u}C_{AB}$ are evaluated at $(\hat f(u,x), g(x))$, $g^{(-1)A}$ is given by \eqref{g^-1}, and $W^{(-2)}_A$ is given by \eqref{W-2}.
\end{theorem}

Recalling the definition of $W_A^{(-2)}$ from \eqref{W-2} and the definition of $g^{(-1)A}$ from \eqref{g(-1)}, we obtain \eqref{final_angular_formula}. 
\begin{proof}
Let $\vphi$ be a BMS transformation. The tensorial nature of curvature tensor implies
\[ R_{arru} = \bar{R} \lt( \vphi_*(\pl_a), \vphi_*(\pl_r), \vphi_*(\pl_r), \vphi_*(\pl_u) \rt).\]
Here $\bar R$ denotes the curvature tensor in $(\bar u, \bar r, \bar x^A)$ coordinate system. By \eqref{tangent_vec_ext_1} and Proposition \ref{curvature}, we obtain 
\begin{align*}
R_{arru} &= K^{-1} \na_a \hat f \bar R_{\bar u \bar r \bar r \bar u} - \frac{1}{r^2} \na_a \hat f g^{(-1)A} \bar R_{\bar u \bar r A \bar u}\\
&\quad - \frac{1}{r^2} \na_a \hat f g^{(-1)A} \bar R_{\bar u A \bar r \bar u} + \frac{1}{r^4} K \na_a \hat f g^{(-1)A} g^{(-1)B} \bar R_{\bar u AB \bar u}\\
&\quad - \frac{1}{r^2} a^{(-1)} \pl_a g^A \bar R_{A\bar u \bar r \bar u} + \frac{1}{r^4} K a^{(-1)} \pl_a g^A g^{(-1)B} \bar R_{A \bar u B \bar u}\\
&\quad + K^{-1} \pl_a g^A \bar R_{A \bar r \bar r \bar u} - \frac{1}{r^2} \pl_a g^A g^{(-1)B} \bar R_{A \bar r B \bar u}\\
&\quad - \frac{1}{r^2} \pl_a g^A g^{(-1)B} \bar R_{AB \bar r \bar u} + \frac{1}{r^4} K \pl_a g^A g^{(-1)B} g^{(-1)D} \bar R_{ABD \bar u} + O(r^{-4}).
\end{align*}
By Proposition \ref{curvature}, we obtain
\begin{align*}
&\frac{3}{2} \lt( g_{ua}^{(-1)} - \pl_a \mathcal{U}^{(-2)} \rt) \\
&= K^2 \lt( 2m - 2 \pl_{\bar u} U^{(-2)}\rt)\na_a \hat f - 2K \pl_{\bar u} W^{(-2)}_A g^{(-1)A} \na_a \hat f\\
&\quad - 2K \pl_{\bar u} W^{(-2)}_A g^{(-1)A} \na_a \hat f + \frac{1}{2} \pl_{\bar u} C_{AB} g^{(-1)A} g^{(-1)B} \na_a \hat f\\
&\quad - \frac{1}{2}K^2 \pl_{\bar u} W^{(-2)}_A \pl_a g^A |\bna \hat f|^2 + \frac{1}{4} K \pl_{\bar u} \pl_{\bar u}C_{AB} \pl_a g^A g^{(-1)B} |\bna\hat f|^2 \\
&\quad + K^2 \pl_a g^A \cdot \frac{3}{2} \lt( \bar g^{(-1)}_{\bar u A} - \pl_A U^{(-2)} \rt) \\
&\quad - K\pl_a g^A g^{(-1)B} \lt( m\sigma_{AB} - \frac{1}{4} C_{BD} \pl_{\bar{u}} C^D_A + \frac{1}{2} \pl_{\bar u} d_{AB} + \frac{1}{2} \na_A W^{(-2)}_B - \frac{1}{2} \na_B W^{(-2)}_A \rt)\\
&\quad - K \pl_a g^A g^{(-1)B} \lt( -\frac{1}{4} C_{BD} \pl_{\bar u}C^D_A + \frac{1}{4} C_{AD} \pl_{\bar u}C^D_B + \na_A W^{(-2)}_B - \na_B W^{(-2)}_A \rt) \\
&\quad + \pl_a g^A g^{(-1)B}g^{(-1)D} \cdot \frac{1}{2} \pl_{\bar u} (\na_A C_{BD} - \na_B C_{AD})
\end{align*}
By Lemma \ref{symmetric_traceless}, the last term is equal to \[ -K^2 \pl_{\bar u} W^{(-2)}_A \pl_a g^A |\bna\hat f|^2 - K \pl_{\bar u} W^{(-2)}_B g^{(-1)B} \na_a \hat f \] and we arrive at
\begin{align*}
\mathfrak{N}_a &=\frac{3}{2}\lt( g_{ua}^{(-1)} + \frac{1}{16} \pl_a \lt( \sigma^{bd}\sigma^{ce} g^{(1)}_{bc} g^{(1)}_{de} \rt) \rt) \\
&= K^2 \lt( 3m - 2 \pl_{\bar u} U^{(-2)}\rt)\na_a \hat f - 3K \pl_{\bar u} W^{(-2)}_A g^{(-1)A} \na_a \hat f \\
&\quad + \frac{1}{2} \pl_{\bar u}\pl_{\bar u}C_{AB} g^{(-1)A} g^{(-1)B} \na_a \hat f - \frac{3}{2} K^2 \pl_{\bar u} W^{(-2)}_A \pl_a g^A |\bna\hat f|^2\\
&\quad + \frac{1}{4} K \pl_{\bar u} \pl_{\bar u}C_{AB} \pl_a g^A g^{(-1)B} |\bna\hat f|^2 + K^2 \pl_a g^A N_A + \frac{1}{8}K^2 \pl_{\bar u}(C_{DE}C^{DE}) \na_a \hat f\\
&\quad -\frac{K}{4}C_{AD}\pl_{\bar u}C^D_B \pl_a g^A g^{(-1)B} + \frac{K}{2} C_{BD} \pl_{\bar u}C^D_A \pl_a g^A g^{(-1)B}\\
&\quad - \frac{3}{2}K\pl_a g^A g^{(-1)B} \lt( \na_A W^{(-2)}_B - \na_B W^{(-2)}_A \rt)
\end{align*}
where we use the fact $d_{AB} = \frac{1}{2}C_A^D C_{DB} = \frac{1}{4} C^{DE}C_{DE} \sigma_{AB}$ for symmetric traceless 2-tensors.

Noting that the second to the last line is equal to \[ -\frac{3K}{4} C_{AD} \pl_{\bar u}C^D_B \pl_a g^A g^{(-1)B} - \frac{K^2}{4}\pl_{\bar u} (C_{DE}C^{DE})\na_a \hat f, \] we obtain
\begin{align*}
\mathfrak{N}_a &= K^2 \pl_a g^A N_A + 3mK^2 \na_a \hat f\\
&\quad + \frac{3}{2}K \pl_a g^A \lt( \na_B W^{(-2)}_A - \na_A W^{(-2)}_B \rt)g^{(-1)B} - \frac{3}{2}K^2 \pl_a g^A \pl_{\bar u} W^{(-2)}_A |\bna\hat f|^2\\
&\quad - \frac{3}{4}K \pl_{\bar u}C^B_D C_{AB} \pl_a g^A g^{(-1)D} + \frac{1}{2} \pl_{\bar u} \pl_{\bar u} C_{AB} g^{(-1)A}g^{(-1)B} \na_a \hat f \\
&\quad  - 3K \pl_{\bar u} W^{(-2)}_A g^{(-1)A} \na_a \hat f + \frac{1}{4} K \pl_{\bar u} \pl_{\bar u} C_{AB} \pl_a g^A g^{(-1)B} |\bna \hat f|^2\\
&\quad - 2K^2 \pl_{\bar u} U^{(-2)} \na_a \hat f - \frac{1}{8} K^2 \pl_{\bar u}(C_{DE}C^{DE}) \na_a \hat f.
\end{align*} 
By \eqref{U-2} the last line vanishes and this completes the proof.
\end{proof}

\section{Deriving the transformation formulae from the metric components}
In the above, we derived the transformation of the mass aspect function and the angular momentum aspect function by reading them off from the curvature tensors. This approach took advantage of the tensorial property and led to a simpler computation in comparison to the traditional approach. For completeness, we present the traditional approach in this section where the transformation formulae were read off from the metric components. In order to do so, we need to expand the metric to further orders. 
\subsection{Setup}
In this subsection, we outline the necessary steps to obtain the transformation formulae directly from metric coefficients. We have
\begin{align*}
\bar r^2 h_{AB} &= \bar r^2 \sigma_{AB}(\bar x) + \bar r C_{AB}(\bar u, \bar x) + \frac{1}{4} (C_{DE}C^{DE})(\bar u, \bar x) \sigma_{AB}(\bar x) + O(\bar r^{-1})\\
&= \lt( K^{-2}r^2 + 2 K^{-1} \rho^{(0)}r + 2K^{-1}\rho^{(-1)} + \rho^{(0)2} \rt) \cdot \\
&\qquad\qquad \lt( \sigma_{AB} + \frac{1}{r}\pl_D \sigma_{AB} g^{(-1)D} + \frac{1}{r^2} \pl_D\sigma_{AB} g^{(-2)D} + \frac{1}{2r^2} \pl_D\pl_E\sigma_{AB} g^{(-1)D} g^{(-1)E}\rt)\\
&\quad + \lt( K^{-1}r + \rho^{(0)}\rt)\lt( C_{AB} + \frac{1}{r}\pl_D C_{AB} g^{(-1)D} + \frac{1}{r} \pl_{\bar u}C_{AB} a^{(-1)} \rt) + \frac{1}{4} (C_{DE}C^{DE})\sigma_{AB} \\
&\quad + O(r^{-1})\\
&= r^2 K^{-2} \sigma_{AB} + r \lt( 2K^{-1} \rho^{(0)} \sigma_{AB} + K^{-1} C_{AB} + K^{-2} \pl_D \sigma_{AB} g^{(-1)D} \rt)\\
&\quad + \lt( 2K^{-1} \rho^{(-1)} + \rho^{(0)2} \rt) \sigma_{AB} + 2K^{-1} \rho^{(0)} \pl_D \sigma_{AB} g^{(-1)D} \\
&\quad + K^{-2}\lt( \pl_D\sigma_{AB} g^{(-2)D} + \frac{1}{2} \pl_D \pl_E \sigma_{AB} g^{(-1)D} g^{(-1)E}\rt) + \rho^{(0)}C_{AB} \\
&\quad + K^{-1} \lt( \pl_D\sigma_{AB} g^{(-1)D} + \pl_{\bar u} C_{AB} a^{(-1)} \rt) + \frac{1}{4} C_{DE} C^{DE} \sigma_{AB} \\
&\quad + O(r^{-1}).
\end{align*}
and \begin{align*}
&d\bar x^A + W^A d\bar u \\
&= \lt( \pl_a g^A + \frac{1}{r}\pl_a g^{(-1)A} + \frac{1}{r^2} \lt( \pl_a g^{(-2)A} + K^2 W^{(-2)A} \na_a \hat f \rt) \rt) dx^a \\
&\quad + \lt( \frac{1}{r} \pl_u g^{(-1)A} + \frac{1}{r^2} \lt( \pl_u g^{(-2)A} + K^3 W^{(-2)A} \rt)  \rt) du\\
&\quad + \frac{1}{r^3} \Bigg( \pl_u g^{(-3)A} + K^4 W^{(-3)A} -2K^4 \rho^{(0)} W^{(-2)A} + K^2 \pl_u a^{(-1)} W^{(-2)A} \\
&\qquad\qquad\qquad\qquad\qquad\qquad + K^3 \pl_D W^{(-2)A} g^{(-1)D} +  K^3 \pl_{\bar u}W^{(-2)A} a^{(-1)} \Bigg)du  \\
&\quad + \lt( -\frac{g^{(-1)A}}{r^2} - \frac{2g^{(-2)A}}{r^3} - \frac{3g^{(-3)A}}{r^4} - \frac{K^2a^{(-1)} W^{(-2)A}}{r^4} \rt)dr
\end{align*}
In the above equations, $\sigma_{AB}, \pl_D\sigma_{AB}, \pl_D \pl_E \sigma_{AB}$ are evaluated at $g(x)$ and $C_{AB}, \pl_D C_{AB}, ...$ are evaluated at $(\hat f(u,x), g(x))$.

Therefore, the metric components in $(u,r,x^a)$ coordinate system are given by 
\begin{align*}
 g_{ra} &= - K^{-1} \na_a \hat f - K^{-2} \sigma_{AB} \pl_a g^A g^{(-1)B} \\
&\quad + \frac{1}{r} \Big[  -K^{-1} \na_a a^{(-1)} - K^{-2} \na_a K a^{(-1)}  \\ &\qquad\qquad  - \lt( 2 K^{-1}\rho^{(0)} \sigma_{AB}  + K^{-1}C_{AB} + K^{-2} \pl_D \sigma_{AB} g^{(-1)D} \rt) \pl_a g^A g^{(-1)B}   \\ 
&\qquad\qquad -K^{-2}\sigma_{AB} \lt( 2 \pl_a g^A g^{(-2)B} + \pl_a g^{(-1)A} g^{(-1)B} \rt) \Big] \\
&\quad + O(r^{-2}),
\end{align*} 
\begin{align*}
g_{rr} &= \frac{1}{r^2} \lt( 2K^{-1} a^{(-1)} + K^{-2} \sigma_{AB} g^{(-1)A} g^{(-1)B} \rt) \\&\quad + \frac{1}{r^3} \Big[ 4K^{-1} a^{(-2)} + \lt( 2K^{-1} \rho^{(0)}\sigma_{AB} + K^{-2} \pl_D \sigma_{AB} g^{(-1)D} + K^{-1}C_{AB} \rt)g^{(-1)A} g^{(-1)B} \\&\qquad\qquad + 4K^{-2} \sigma_{AB} g^{(-1)A}g^{(-2)B} \Big]\\
&\quad + O(r^{-4}),
\end{align*}
\begin{align*}
g_{ab} &= r^2 \sigma_{ab} \\
&\quad + r \Big[ K^{-2} \na_a K \na_b \hat f + K^{-2} \na_b K \na_a \hat f + 2K\rho^{(0)}\sigma_{ab} \\
&\quad\qquad +  K^{-2} \sigma_{AB} \lt(\pl_a g^A \pl_b g^{(-1)B} + \pl_a g^{(-1)A} \pl_b g^B \rt) \\
&\quad\qquad+  \lt( K^{-2} \pl_D \sigma_{AB} g^{(-1)D} + K^{-1} C_{AB} \rt) \pl_a g^A \pl_b g^B \Big] \\
&\quad + g^{(0)}_{ab} + O(r^{-1}),
\end{align*}
where
\begin{align*}
g^{(0)}_{ab} &= - \na_a \hat f \na_b \hat f - \na_a \hat f \na_b \rho^{(0)} - \na_b \hat f \na_a \rho^{(0)} + K^{-2} \na_a a^{(-1)} \na_b K + K^{-2} \na_b a^{(-1)} \na_a K \\
&\quad + \lt( 2 K^{-1} \rho^{(-1)} + \rho^{(0)2} \rt)\bar\sigma_{ab}\\
&\quad + \lt( K^{-2} \pl_D \sigma_{AB} g^{(-2)D} + \frac{1}{2} K^{-2} \pl_D \pl_E \sigma_{AB} g^{(-1)D} g^{(-1)E} \rt) \pl_a g^A \pl_b g^B \\
&\quad + \lt( K^{-1} \pl_D C_{AB} g^{(-1)D} + \frac{1}{4} C_{DE}C^{DE} \sigma_{AB} +K^{-1} \pl_{\bar u}C_{AB} a^{(-1)} \rt) \pl_a g^A \pl_b g^B\\
&\quad + K^{-2} \sigma_{AB} \lt( \pl_a g^A \pl_b g^{(-2)B} + \pl_b g^B \pl_a g^{(-2)A} + \pl_a g^{(-1)A} \pl_b g^{(-1)B} \rt)\\
&\quad + \sigma_{AB} \lt( \pl_a g^A W^{(-2)B} \na_b \hat f + \pl_b g^B W^{(-2)A} \na_a \hat f \rt)\\ 
&\quad +2 \rho^{(0)} K^{-1}\pl_D \sigma_{AB} g^{(-1)D} \pl_a g^A \pl_b g^B + \rho^{(0)} C_{AB}\pl_a g^A \pl_b g^B \\
&\quad + 2K^{-1} \rho^{(0)} \sigma_{AB} \lt( \pl_a g^A \pl_b g^{(-1)B} + \pl_b g^B \pl_a g^{(-1)A} \rt)\\
&\quad + \lt( K^{-2} \pl_D \sigma_{AB} g^{(-1)D} + K^{-1} C_{AB} \rt)\lt( \pl_a g^A \pl_b g^{(-1)B} + \pl_b g^B \pl_a g^{(-1)A} \rt).
\end{align*}
In addition to $g^{(-1)A}, a^{(-1)}, \rho^{(0)}$, we will solve inductively 
\begin{enumerate}
\item $g^{(-2)A}$ from $g_{ra}^{(-1)}=0$,
\item $\rho^{(-1)}$ from 
\begin{align}\label{gab0}
g^{(0)}_{ab} = \frac{1}{2} \sigma^{cd} g^{(1)}_{ac} g^{(1)}_{bd}, 
\end{align} 
\item $a^{(-2)}$ from $g_{rr}^{(-3)}=0,$ and
\item $g^{(-3)A}$ from $g_{ra}^{(-2)}=0.$
\end{enumerate} 
We remind the readers that \eqref{gab0} comes from the outgoing radiation condition. See Madler-Winicour \cite[page 13-14]{MW}.

\subsection{Mass aspect function}
To get the transformation formula by regarding the mass aspect function as a metric coefficient, one needs to compute $g^{(-2)B}$ and $\rho^{(-1)}$ in the expansions of a BMS transformation.
\begin{lem}\label{g(-2)}
\begin{align*}
g^{(-2)B} &= -K \rho^{(0)} g^{(-1)B}- \frac{1}{2} \Gamma_{ED}^B  g^{(-1)E} g^{(-1)D}  - \frac{1}{2} K \sigma^{BD}C_{DE} g^{(-1)E}  \end{align*}
\end{lem}
\begin{proof}
The equation $g_{rr}^{(-2)}=0$ implies
\[ 2K^{-1} a^{(-1)} + K^{-2} \sigma_{AB} g^{(-1)A} g^{(-1)B} =0 \]
and
\begin{align*}
-K^{-1} \na_a a^{(-1)} - K^{-2} \na_a K a^{(-1)} &= -K^{-2} \na_a \lt( K a^{(-1)} \rt) \\
&= K^{-2} \lt( \frac{1}{2} \pl_a g^D \pl_D \sigma_{AB} g^{(-1)A} g^{(-1)B} + \sigma_{AB} \pl_a g^{(-1)A} g^{(-1)B} \rt)
\end{align*}
Hence, we have
\begin{align*}
g_{ra}^{(-1)} &= - \lt( 2 K^{-1}\rho^{(0)} \sigma_{AB}  + K^{-1}C_{AB} + K^{-2} \pl_D \sigma_{AB} g^{(-1)D} \rt) \pl_a g^A g^{(-1)B}   \\ 
&\qquad\quad + \frac{1}{2} K^{-2} \pl_A \sigma_{DB} \pl_a g^A g^{(-1)D} g^{(-1)B} - 2 K^{-2}\sigma_{AB} \pl_a g^A g^{(-2)B} 
\end{align*}
The equation $g_{ra}^{(-1)}=0$ becomes
\begin{align*}
\sigma_{AB} \pl_a g^A g^{(-2)B} &= -K \rho^{(0)} \sigma_{AB} \pl_a g^A g^{(-1)B} - \frac{1}{2} K C_{AB} \pl_a g^A g^{(-1)B}\\
&\quad + \lt( \frac{1}{4} \pl_A\sigma_{DB} - \frac{1}{2}\pl_D \sigma_{AB} \rt) \pl_a g^A  g^{(-1)D} g^{(-1)B} 
\end{align*}
and it is straightforward to check $g^{(-2)B}$ has the claimed form.
\end{proof}

\begin{defn}
We denote the pull-back of shear tensor by
\begin{align}
\mathcal{C}_{ab} (u,x) = C_{AB} \lt( \hat f(u,x), g(x) \rt)\pl_a g^A(x) \pl_b g^B(x).
\end{align}
\end{defn}

The derivative of $\mathcal{C}_{ab}$ is computed in the proof of Theorem \ref{modified}
\begin{lem}\label{derivative of shear tensor}
The derivative of shear tensors are related by
\begin{align*}
\bna_a \mathcal{C}_{de}(u,x) = \lt( \pl_{\bar u}C_{DE} \na_a \hat f + \na_A C_{DE} \pl_a g^A \rt)\pl_d g^D \pl_e g^E.
\end{align*}
Here $\pl_{\bar u}C_{DE}$ and $\na_A C_{DE}$ are evaluated at $\lt( \hat f(u,x), g(x)\rt).$
\end{lem}

\begin{lem}\label{rho-1}
\begin{align*}
\rho^{(-1)} &= \frac{1}{2}K |\bna \hat f|^2 + \frac{1}{4} K \lt( \bna_a \bna_b \hat f - \frac{1}{2} \bar\Delta \hat f \bar\sigma_{ab} \rt) \lt( \bna^a \bna^b \hat f - \frac{1}{2} \bar\Delta\hat f \bar\sigma^{ab} \rt) \\
&\quad - \frac{1}{4}K^{-1} \pl_{\bar u}C_{AB} g^{(-1)A}g^{(-1)B} + \sigma_{AB} W^{(-2)A} g^{(-1)B} - \frac{1}{4}K \mathcal{C}_{ab}\bna^a \bna^b \hat f .
\end{align*}
Here $\pl_{\bar u}C_{DE}$ and $W^{(-2)A}$ are evaluated at $\lt( \hat f(u,x), g(x)\rt).$
\end{lem}
\begin{proof}
In the proof we take normal coordinates at $x_0$, $g(x_0)$ and evaluate all tensors at $x_0.$ This implies $\bar\Gamma_{ab}^c(x_0)= \Gamma_{AB}^C(g(x_0)) = \pl_a\pl_b g^A(x_0)=0.$ We write $g_{ab}^{(0)} = 2K^{-1}\rho^{(-1)} \bar\sigma_{ab} + \RN{1} + \RN{2}$ where terms involving $C_{AB}$ are grouped in $\RN{2}$. 

We have
\begin{align*}
\RN{1} &= - \na_a \hat f \na_b \hat f \underline{- \na_a \hat f \na_b \rho^{(0)} - \na_b \hat f \na_a \rho^{(0)} } + K^{-2} \na_a a^{(-1)} \na_b K + K^{-2} \na_b a^{(-1)} \na_a K \\
&\quad + \rho^{(0)2} \bar\sigma_{ab} + \frac{1}{2} K^{-2} \pl_D \pl_E \sigma_{AB} g^{(-1)D} g^{(-1)E}  \pl_a g^A \pl_b g^B \\
&\quad + K^{-2} \sigma_{AB}  \pl_a g^A \lt( \underline{ \pl_b \lt( -K \rho^{(0)} g^{(-1)B} \rt)} + \pl_b \lt( - \frac{1}{2} \Gamma^B_{DE} \rt) g^{(-1)D}g^{(-1)E} \rt) \\
&\quad + K^{-2}\sigma_{AB} \pl_b g^B \lt( \underline{ \pl_a \lt( -K \rho^{(0)} g^{(-1)A} \rt) } + \pl_a \lt( -\frac{1}{2}\Gamma^A_{DE}\rt) g^{(-1)D}g^{(-1)E} \rt)\\
&\quad + K^{-2} \sigma_{AB} \pl_a g^{(-1)A}\pl_b g^{(-1)B} + \underline{ 2K^{-1} \rho^{(0)} \sigma_{AB} \lt( \pl_a g^A \pl_b g^{(-1)B} + \pl_b g^B \pl_a g^{(-1)A} \rt) }.
\end{align*}
By Lemma \ref{g(-1)}, the underlined terms become $-2 \rho^{(0)} \bna_a \bna_b \hat f$. By Lemma \ref{g(-1)} and Lemma \ref{rho(0)}, we have
\begin{align*}
K^{-2} \na_a a^{(-1)} \na_b K + K^{-2} \na_b a^{(-1)} \na_a K + K^{-2} \sigma_{AB} \pl_a g^{(-1)A}\pl_b g^{(-1)B}
= \bna_a \bna^c \hat f \bna_b \bna_c \hat f.
\end{align*}
Finally, from $\pl_b \Gamma^B_{DE} = \pl_b g^F \pl_F \Gamma^B_{DE}$ and $\pl_D\pl_E\sigma_{AB} = \pl_D \Gamma^F_{EA} \sigma_{FB} + \pl_D \Gamma^F_{EB}\sigma_{FA}$, terms involving second derivative of $\sigma$ can be simplified
\begin{align*}
&\frac{1}{2} K^{-2} \lt( \pl_D \pl_E \sigma_{AB} \pl_a g^A \pl_b g^B - \sigma_{AB} \lt( \pl_a g^A \pl_b  \Gamma^B_{DE}  + \pl_b g^B \pl_a \Gamma^A_{DE} \rt) \rt) g^{(-1)D}g^{(-1)E}\\
&= -\frac{1}{2} K^{-2}\lt( \pl_F \Gamma^B_{DE} - \frac{1}{2}\pl_D \Gamma^B_{FE} - \frac{1}{2} \pl_E \Gamma^B_{DF} + \pl_F \Gamma^A_{DE} - \frac{1}{2}\pl_D \Gamma^A_{FE} - \frac{1}{2}\pl_E \Gamma^A_{DF} \rt) g^{(-1)D}g^{(-1)E} \\
&= -\frac{1}{2} \lt( R_{FD\;\;E}^{\;\;\;\;B} + R_{FD\;\;E}^{\;\;\;\;A} \rt) \pl_d g^D \pl_e g^E \bna^d \hat f \bna^e \hat f\\
&= - |\bna \hat f|^2 \bar\sigma_{ab} + \na_a \hat f \na_b \hat f.  
\end{align*}
Putting these together, we obtain
\[  \RN{1} = -|\bna \hat f|^2 \bar\sigma_{ab} + \lt( \bna_a \bna^c \hat f - \frac{1}{2} \bna\Delta \hat f \delta_a^c \rt)\lt( \bna_b \bna_c \hat f - \frac{1}{2}\Delta \hat f \bar\sigma_{bc} \rt).\]

Next, we simplify $\RN{2}$. By Lemma \ref{derivative of shear tensor}, we have 
\begin{align*}
\RN{2} &= \lt( K^{-1} \na_D C_{AB} g^{(-1)D} + K^{-1} \pl_{\bar u} C_{AB} a^{(-1)} \rt) \pl_a g^A \pl_b g^B\\
&\quad + K^{-2} \sigma_{AB} \pl_a g^A \pl_b \lt( -\frac{1}{2} K \sigma^{BD} C_{DE} g^{(-1)E} \rt) + K^{-2} \sigma_{AB} \pl_b g^B \pl_a \lt( -\frac{1}{2} K \sigma^{AD} C_{DE} g^{(-1)E} \rt) \\
&\quad + \sigma_{AB} \lt( \pl_a g^A W^{(-2)B} \na_b \hat f + \pl_b g^B W^{(-2)A} \na_a \hat f \rt) \\
&\quad + K^{-1} C_{AB} \lt( \pl_a g^A \pl_b g^{(-1)B} + \pl_b g^B \pl_a g^{(-1)A} \rt) \\
&\quad + \frac{1}{4} C_{DE}C^{DE} \bar\sigma_{ab} + \rho^{(0)} \mathcal{C}_{ab}\\
&= -\bar\na_d \mathcal{C}_{ab} \bna^d \hat f + \frac{1}{2} \pl_{\bar u}C_{AB} |\bna \hat f|^2 \pl_a g^A \pl_b g^B\\
&\quad + K^{-1} \na_b K \mathcal{C}_{ae} \bna^e \hat f + \frac{1}{2} \bna_b \mathcal{C}_{ae} \bna^e \hat f + \frac{1}{2} \mathcal{C}_{ae} \bna_b \bna^e \hat f \\ &\quad + K^{-1} \na_a K \mathcal{C}_{be} \bna^e \hat f + \frac{1}{2} \bna_a \mathcal{C}_{be} \bna^e \hat f + \frac{1}{2} \mathcal{C}_{be} \bna_a \bna^e \hat f \\
&\quad + \frac{1}{2} \bna^d \mathcal{C}_{ad} \na_b \hat f + \frac{1}{2} \bna^d \mathcal{C}_{bd} \na_a \hat f - \frac{1}{2} \pl_{\bar u}C_{AD} \bna^d \hat f \pl_d g^D \pl_a g^A \na_b \hat f - \frac{1}{2} \pl_{\bar u} C_{BD} \bna^d \hat f \pl_d g^D \pl_b g^B \na_a \hat f\\
&\quad - K^{-1} \na_b K \mathcal{C}_{ae} \bna^e \hat f - \mathcal{C}_{ae} \bna_b \bna_e \hat f - K^{-1} \na_a K \mathcal{C}_{be} \bna^e \hat f - \mathcal{C}_{be} \bna_a \bna^e \hat f\\
&\quad + \frac{1}{4} C_{DE}C^{DE} \bar\sigma_{ab} + \rho^{(0)} \mathcal{C}_{ab}  
\end{align*}

By Lemma \ref{symmetric_traceless}, the terms involving covariant derivative of $\mathcal{C}$ is a multiple of $\bar\sigma_{ab}:$
\begin{align*}
\lt( - \bna_d \mathcal{C}_{ab} + \frac{1}{2} \bna_b \mathcal{C}_{ad} + \frac{1}{2} \bna_a \mathcal{C}_{bd} + \frac{1}{2} \bna^e \mathcal{C}_{ae} \bar\sigma_{bd} + \frac{1}{2} \bna^e \mathcal{C}_{be} \bar\sigma_{ad} \rt) \bna^d \hat f = \bna^e \mathcal{C}_{de} \bna^d \hat f \bar\sigma_{ab}.
\end{align*}
Hence, we obtain
\begin{subequations}
\begin{align}
\RN{2} &= \bna^e \mathcal{C}_{de} \bna^d \hat f \bar\sigma_{ab} + \frac{1}{4}C_{DE}C^{DE} \bar\sigma_{ab}\notag\\
&\quad + \frac{1}{2} \lt( \pl_{\bar u} C_{AB} \pl_a g^A \pl_b g^B |\bna \hat f|^2 - \pl_{\bar u} C_{AD} \pl_a g^A \pl_d g^D \na_b \hat f \bna^d \hat f - \pl_{\bar u}C_{BD} \pl_b g^B \pl_d g^D \na_a \hat f \bna^d \hat f \rt) \label{g0a}\\
&\quad + \rho^{(0)} \mathcal{C}_{ab} - \frac{1}{2} \mathcal{C}_{ae} \bna_b \bna^e \hat f - \frac{1}{2} \mathcal{C}_{be} \bna_a \bna^e \hat f \label{g0b}. 
\end{align}
\end{subequations}
Applying Lemma \ref{multiple of metric} to $T_{ab} = \pl_{\bar u}C_{AB} \pl_a g^A \pl_b g^B, S_{ab} = \na_a \hat f \na_b \hat f$ and $T_{ab} = \mathcal{C}_{ab}, S_{ab} = \bna_a \bna_b \hat f,$ we infer that \eqref{g0a} and \eqref{g0b}  are both a multiple of $\bar\sigma_{ab}.$ 

We conclude that
\begin{align*}
\bar\sigma^{ab}g^{(0)}_{ab} &= 4K^{-1}\rho^{(-1)} - 2 |\bar\na \hat f|^2 +  \bna^a \bna^b \hat f \bna_a\bna_b \hat f - \frac{1}{2} \lt( \bar\Delta \hat f \rt)^2\\
&\quad + 2 \bna^e \mathcal{C}_{de} \bna^d \hat f + \frac{1}{2} C_{DE}C^{DE} - \pl_{\bar u} C_{AB} \pl_a g^A \pl_b g^B \bna^a \hat f \bna^b \hat f - \mathcal{C}_{ab} \bna^a \bna^b \hat f
\end{align*}
On the other hand, we have
\begin{align*}
\frac{1}{2} \bar\sigma^{ab} \sigma^{cd} g^{(1)}_{ac} g^{(1)}_{bd} = \frac{1}{2} \bar\sigma^{ab} \bar\sigma^{cd} \mathcal{C}_{ac}\mathcal{C}_{bd}  - 2 \mathcal{C}_{ab} \bna^a \bna^b \hat f + 2 \bna_a \bna_b \hat f \bna^a \bna^b \hat f - \lt( \bar\Delta \hat f\rt)^2.
\end{align*} From $\bar\sigma^{ab}g_{ab}^{(0)} = \frac{1}{2}\bar\sigma^{ab} \sigma^{cd} g_{ac}^{(1)}g_{bd}^{(1)}$ we get
\begin{align*}
\rho^{(-1)} &= \frac{1}{2}K |\bna \hat f|^2 + \frac{1}{4} K \lt( \bna_a \bna_b \hat f - \frac{1}{2} \bar\Delta \hat f \bar\sigma_{ab} \rt) \lt( \bna^a \bna^b \hat f - \frac{1}{2} \bar\Delta\hat f \bar\sigma^{ab} \rt) \\
&\quad -\frac{1}{2} K \bna^d \mathcal{C}_{de} \bna^d \hat f + \frac{1}{4} K \pl_{\bar u}C_{AB} \pl_a g^A \pl_b g^B \bna^a \hat f \bna^b \hat f - \frac{1}{4} K \mathcal{C}_{ab} \bna^a \bna^b \hat f\\
&= \frac{1}{2}K |\bna \hat f|^2 + \frac{1}{4} K \lt( \bna_a \bna_b \hat f - \frac{1}{2} \bar\Delta \hat f \bar\sigma_{ab} \rt) \lt( \bna^a \bna^b \hat f - \frac{1}{2} \bar\Delta\hat f \bar\sigma^{ab} \rt) \\
&\quad - \frac{1}{4}K^{-1} \pl_{\bar u}C_{AB} g^{(-1)A}g^{(-1)B} + \sigma_{AB} W^{(-2)A} g^{(-1)B} - \frac{1}{4}K \mathcal{C}_{ab}\bna^a \bna^b \hat f 
\end{align*}
where we used Lemma \ref{derivative of shear tensor} in the second equality.
\end{proof}

\begin{remark}
In the proof we showed that $g^{(0)}_{ab}$ is proportional to $\sigma_{ab}$. This is consistent with the Einstein equations and the outgoing radiation condition.
\end{remark}

We are ready to compute the transformation of the mass aspect function. We have
\begin{align*}
\pl_u \rho^{(-1)} &= K \bna^c K \na \hat f -\frac{1}{2} K \pl_u \lt( \bna^e \mathcal{C}_{de} \rt)\bna^d \hat f - \frac{1}{2} K \bna^e \mathcal{C}_{de} \bna^d K \\
&\quad + \frac{1}{4} K^2 \pl_{\bar u} \pl_{\bar u}C_{AB} \pl_a g^A \pl_b g^B \bna^a \hat f \bna^b \hat f + \frac{1}{2} K \pl_{\bar u} C_{AB} \pl_a g^A \pl_b g^B \bna^a K \bna^b \hat f \\
&\quad - \frac{1}{4} K \pl_u \mathcal{C}_{ab} \bna^a \bna^b \hat f \end{align*}
The mass aspect function by definition is the coefficient $g_{uu}^{(-1)}$ and transforms as
\begin{align*}
\mathfrak{m}&= mK^3 - K \pl_u a^{(-1)} - \pl_u a^{(-1)} \pl_u \rho^{(0)} \\
&\quad - K \pl_u \rho^{(-1)} + K^{-2} \sigma_{AB} \pl_u g^{(-1)A} \pl_u g^{(-2)B} + K\sigma_{AB} \pl_ug^{(-1)A} W^{(-2)B} \\
&\quad + K^{-1} \rho^{(0)} \sigma_{AB} \pl_u g^{(-1)A} \pl_u g^{(-1)B} \\
&\quad + \frac{1}{2} \lt( K^{-2}  \pl_D \sigma_{AB} g^{(-1)D} + K^{-1} C_{AB}\rt) \pl_u g^{(-1)A} \pl_u g^{(-1)B}\\
&= mK^3 + \frac{1}{2}K^2 \pl_u \bna^d \mathcal{C}_{de} \bna^e \hat f - \frac{1}{4} K^3 \pl_{\bar u}\pl_{\bar u} C_{AB} \pl_a g^A \pl_b g^B \bna^a \hat f \bna^b \hat f  \\
&\quad + \frac{1}{4}K^2 \pl_u \mathcal{C}_{ab} \bna^a \bna^b \hat f  -\frac{1}{2} K^2 \pl_{\bar u} C_{AB}  \pl_a g^A \pl_b g^B   \bna^a K \bna^b \hat f \\
&= mK^3 + \frac{1}{4} K^3 \pl_{\bar u}\pl_{\bar u} C_{AB} \pl_a g^A \pl_b g^B \bna^a \hat f \bna^b \hat f \\
&\quad + K^3 \pl_{\bar u} W^{(-2)}_A \pl_a g^A \bna^a \hat f + \frac{1}{4}K^3 \pl_{\bar u}C_{AB} \pl_a g^A \pl_b g^B \bna^a\bna^b \hat f \end{align*}
where we used Lemma \ref{derivative of shear tensor} in the last equality. This recovers Theorem \ref{mass aspect}. 

\subsection{Angular momentum aspect}
To get the transformation formula by regarding angular momentum aspect as a metric coefficient, one needs to further compute $a^{(-2)}$ and $g^{(-3)A}$. First of all, $g_{rr}^{(-3)}=0$ and Lemma \ref{g(-2)} give $a^{(-2)}$.
\begin{lem}\label{a-2}
\begin{align*}
a^{(-2)} = \frac{1}{2} \rho^{(0)} \sigma_{AB} g^{(-1)A} g^{(-1)B} + \frac{1}{4} C_{AB} g^{(-1)A} g^{(-1)B}
\end{align*}
\end{lem}

Next we solve $g^{(-3)A}$ from $g_{ra}^{(-2)}=0$. Since we only need $\pl_u g^{(-3)B}$ later, it suffices to work in a normal coordinate. 
\begin{lem}\label{g-3} In a normal coordinate with $\Gamma^A_{BC} =0$ at $g(x_0)$, we have
\begin{align*}
g^{(-3)B}(x_0) &= \lt( -K\rho^{(-1)} - \frac{1}{12}\pl_{\bar u}C_{DE}g^{(-1)D}g^{(-1)E} + \frac{1}{16}K^2C_{DE}C^{DE} + \frac{1}{2}K\na^A C_{AD} g^{(-1)D}\rt) g^{(-1)B} \\
&\quad -\frac{1}{4}K \na^B C_{DE}g^{(-1)D}g^{(-1)E} \\ &\quad  - \frac{1}{4} K^3 \na^A C^B_A |\bna \hat f|^2  + K^2\rho^{(0)}C^B_D g^{(-1)D} + \frac{1}{6}K^2 \pl_{\bar u}C^B_D g^{(-1)D} |\bna \hat f|^2    \\
&\quad + \frac{1}{6}K^2|\bna\hat f|^2 g^{(-1)B} + K^2\rho^{(0)2}g^{(-1)B} - \frac{1}{6}\pl_D\Gamma^B_{EF}g^{(-1)D}g^{(-1)E}g^{(-1)F}.
\end{align*}
\end{lem}
\begin{proof}
 We have
\begin{align*}
&g_{ra}^{(-2)} \\
&= a^{(-1)} \na_a \hat f + \rho^{(-1)} \na_a \hat f - K^{-1} \na_a a^{(-2)} + a^{(-1)} \pl_a \rho^{(0)} - 2K^{-2} a^{(-2)} \na_a K - K U^{(-2)} \na_a \hat f\\
&\quad - K^{-2} \sigma_{AB} \lt( 3\pl_a g^A g^{(-3)B} + K^2 a^{(-1)} \pl_a g^A W^{(-2)B} + 2 \pl_a g^{(-1)A} g^{(-2)B} \rt)\\
&\quad -K^{-2}\sigma_{AB} \lt( \pl_a g^{(-2)A} g^{(-1)B} + K^2 W^{(-2)A} g^{(-1)B} \pl_a \hat f \rt)\\
&\quad -\lt( 2K^{-1}\rho^{(0)} \sigma_{AB} + K^{-1} C_{AB} + K^{-2} \pl_D \sigma_{AB} g^{(-1)D} \rt) \lt( 2\pl_a g^A g^{(-2)B} + \pl_a g^{(-1)A} g^{(-1)B} \rt)\\
&\quad -  \lt( 2K^{-1} \rho^{(-1)} + \rho^{(0)2} \rt) \sigma_{AB} \pl_a g^A g^{(-1)B}\\
&\quad -\pl_a g^A g^{(-1)B} \lt( 2 K^{-1} \rho^{(0)}\pl_D \sigma_{AB} g^{(-1)D} + K^{-2} ( \pl_D \sigma_{AB} g^{(-2)D} + \frac{1}{2} \pl_D\pl_E\sigma_{AB} g^{(-1)D} g^{(-1)E} )\rt)\\
&\quad - \pl_a g^A g^{(-1)B} \lt( \rho^{(0)}C_{AB} + K^{-1} \lt( \pl_D C_{AB} g^{(-1)D} + \pl_{\bar u} C_{AB} a^{(-1)} \rt) + \frac{1}{4} C_{DE}C^{DE}\sigma_{AB} \rt) \\
&= -3K^{-2} \sigma_{AB} \pl_a g^A g^{(-3)B} + 3 \rho^{(-1)}\na_a \hat f + \RN{1} + \RN{2},
\end{align*}
where we group terms containing $C_{AB}$ in $\RN{1}.$ We have
\begin{align*}
\RN{1} &= - \frac{1}{4} K^{-1} \lt( \pl_{\bar u} C_{AB} \na_a \hat f + \na_D C_{AB} \pl_a g^D \rt) g^{(-1)A} g^{(-1)B} + \frac{1}{16} K C_{DE} C^{DE} \na_a \hat f\\
&\quad + \frac{1}{4} K \na^B C_{AB} |\bna \hat f|^2 \pl_a g^A \\
&\quad + \frac{1}{2} K^{-1} \lt( \pl_{\bar u} C_{BE} \na_a \hat f + \na_D C_{BE} \pl_a g^D  \rt) g^{(-1)B} g^{(-1)E} - \frac{1}{2} \na^A C_{AB} \na_a \hat f g^{(-1)B}\\
&\quad + 4\rho^{(0)} C_{AB}  \pl_a g^A g^{(-1)B}  + \frac{1}{2} C_{DE} C^{DE} \sigma_{AB} \pl_a g^A g^{(-1)B} \\
&\quad -\rho^{(0)}C_{AB} \pl_a g^A g^{(-1)B} - K^{-1} \lt( \na_D C_{AB} g^{(-1)D} + \pl_{\bar u} C_{AB} a^{(-1)} \rt)\pl_a g^A g^{(-1)B} \\
&\quad - \frac{1}{4} C_{DE} C^{DE} \sigma_{AB} \pl_a g^A g^{(-1)B}\\
&= \frac{1}{4}K^{-1} \pl_{\bar u} C_{AB} g^{(-1)A} g^{(-1)B} \na_a \hat f - \frac{3}{16} K C_{DE} C^{DE} \na_a \hat f \\
&\quad + \frac{1}{4} K \pl_a g^A \na^B C_{AB} |\bna \hat f|^2 - \frac{1}{2} \na^B C_{AB} g^{(-1)A} \na_a \hat f + 3 \rho^{(0)} C_{AB} \pl_a g^A g^{(-1)B}\\
&\quad + \frac{1}{2} \pl_{\bar u}C_{AB} \pl_a g^A g^{(-1)B} |\bna\hat f|^2 \\
&\quad + K^{-1} g^{(-1)B} g^{(-1)D} \pl_a g^A \lt( \frac{1}{4} \na_A C_{DB} - \na_D C_{AB} \rt),
\end{align*}
We rearrange the last term by Lemma \ref{symmetric_traceless} \[ -\na_D C_{AB} g^{(-1)B} g^{(-1)D} = \lt( -\na_A C_{BD} - \na^E C_{AE} \sigma_{BD} + \na^E C_{DE} \sigma_{AB}\rt) g^{(-1)B}g^{(-1)D} \]
to obtain
\begin{align*}
\RN{1} &=  \frac{1}{4}K^{-1} \pl_{\bar u} C_{AB} g^{(-1)A} g^{(-1)B} \na_a \hat f - \frac{3}{16} K C_{DE} C^{DE} \na_a \hat f \\
&\quad - \frac{3}{4} K \pl_a g^A \na^B C_{AB} |\bna \hat f|^2 - \frac{3}{2} \na^B C_{AB} g^{(-1)A} \na_a \hat f + 3 \rho^{(0)} C_{AB} \pl_a g^A g^{(-1)B}\\
&\quad + \frac{1}{2} \pl_{\bar u}C_{AB} \pl_a g^A g^{(-1)B} |\bna\hat f|^2\\
&\quad - \frac{3}{4}K^{-1}\na_A C_{BD}\pl_a g^A g^{(-1)B}g^{(-1)D}.
\end{align*}
Next, we have
\begin{align*}
\RN{2} &= a^{(-1)} \na_a \hat f - K^{-1} \na_a \lt( \frac{1}{2}\rho^{(0)} \sigma_{AB} g^{(-1)A} g^{(-1)B}\rt) + a^{(-1)} \na_a \rho^{(0)} - K^{-2} \rho^{(0)}\sigma_{AB} g^{(-1)A} g^{(-1)B} \na_a K\\
&\quad + 2K^{-1} \rho^{(0)} \sigma_{AB} \pl_a g^{(-1)A} g^{(-1)B} + K^{-2}\sigma_{AB} \pl_a \lt( K\rho^{(0)} g^{(-1)A} + \frac{1}{2}\Gamma^A_{ED} g^{(-1)E} g^{(-1)D} \rt)g^{(-1)B} \\
&\quad + 4\rho^{(0)2}\sigma_{AB}\pl_a g^A g^{(-1)B} -2K^{-1} \rho^{(0)}\sigma_{AB} \pl_a g^{(-1)A} g^{(-1)B} - \rho^{(0)2}\sigma_{AB} \pl_a g^A g^{(-1)B}\\
&\quad - \frac{1}{2}K^{-2} \pl_D\pl_E \sigma_{AB} \pl_a g^A g^{(-1)B} g^{(-1)D} g^{(-1)E}.
\end{align*} 
We simplify the terms involving second derivatives of metric 
\begin{align*}
&\frac{1}{2}K^{-2} \lt( \sigma_{FB} \pl_A \Gamma^F_{ED} - \pl_D\pl_E\sigma_{AB}\rt)\pl_a g^A g^{(-1)B}g^{(-1)D}g^{(-1)E} \\
&= \frac{1}{2} K^{-1} \lt( R_{ADBE} -\pl_D \Gamma^F_{EB}\sigma_{AF}\rt) \pl_a g^A g^{(-1)B}g^{(-1)D}g^{(-1)E} 
\end{align*}
to obtain
\begin{align*}
\RN{2} &= -\frac{1}{2}K|\bna\hat f|^2 \na_a \hat f + 3 \rho^{(0)2}\sigma_{AB} \pl_a g^A g^{(-1)B} - \frac{1}{2} K^{-2} \pl_D \Gamma^F_{EB} \sigma_{AF} g^{(-1)B} g^{(-1)D} g^{(-1)E} \pl_a g^A.
\end{align*}
Hence, $g_{ra}^{(-2)}=0$ implies
\begin{align*}
&K^{-2}\sigma_{AB}\pl_a g^A g^{(-3)B} \\
&=  \rho^{(-1)} \na_a \hat f + \frac{1}{12}K^{-1} \pl_{\bar u} C_{AB} g^{(-1)A} g^{(-1)B} \na_a \hat f - \frac{1}{16} K C_{DE} C^{DE} \na_a \hat f \\
&\quad -\frac{1}{4}K^{-1} \na_A C_{BD} \pl_a g^A g^{(-1)B}g^{(-1)D} - \frac{1}{4} K \pl_a g^A \na^B C_{AB} |\bna \hat f|^2\\
&\quad - \frac{1}{2}\na^B C_{AB} g^{(-1)A}\na_a \hat f + \rho^{(0)} C_{AB}\pl_a g^A g^{(-1)B}\\
&\quad + \frac{1}{6} \pl_{\bar u}C_{AB} \pl_a g^A g^{(-1)B} |\bna \hat f|^2 + \frac{1}{2}K^{-1} \Gamma_{DE}^F \sigma_{AF} C_B^E \pl_a g^A g^{(-1)B}g^{(-1)D}\\
&\quad -\frac{1}{6}K|\bna\hat f|^2 \na_a \hat f +  \rho^{(0)2}\sigma_{AB} \pl_a g^A g^{(-1)B} - \frac{1}{6} K^{-2} \pl_D \Gamma^F_{EB} \sigma_{AF} g^{(-1)B} g^{(-1)D} g^{(-1)E} \pl_a g^A
\end{align*}
and we obtain
\begin{align*}
g^{(-3)B} &= \lt( -K\rho^{(-1)} - \frac{1}{12}\pl_{\bar u}C_{DE}g^{(-1)D}g^{(-1)E} + \frac{1}{16}K^2C_{DE}C^{DE} + \frac{1}{2}K\na^A C_{AD} g^{(-1)D}\rt) g^{(-1)B} \\
&\quad -\frac{1}{4}K \na^B C_{DE}g^{(-1)D}g^{(-1)E} \\ &\quad  - \frac{1}{4} K^3 \na^A C^B_A |\bna \hat f|^2  + K^2\rho^{(0)}C^B_D g^{(-1)D} + \frac{1}{6}K^2 \pl_{\bar u}C^B_D g^{(-1)D} |\bna \hat f|^2    \\
&\quad + \frac{1}{6}K^2|\bna\hat f|^2 g^{(-1)B} + K^2\rho^{(0)2}g^{(-1)B} - \frac{1}{6}\pl_D\Gamma^B_{EF}g^{(-1)D}g^{(-1)E}g^{(-1)F}.
\end{align*}
This completes the proof.
\end{proof}

The next proposition encodes the transformation of angular momentum aspect and recovers Theorem \ref{angular momentum aspect} by
\begin{align*}
N_A = \frac{3}{2} \lt( \bar g_{\bar u A}^{(-1)} + \frac{1}{16} \pl_A(C_{DE}C^{DE}) \rt).
\end{align*}
\begin{prop}
\begin{align*}
&g_{ua}^{(-1)} + \frac{1}{16} \na_a \lt( \sigma^{bd}\sigma^{ce} g^{(1)}_{bc} g^{(1)}_{de} \rt)\\
&= K^2 \pl_a g^A \lt( \bar g^{(-1)}_{\bar u A} + \frac{1}{16}\na_A \lt( C_{DE}C^{DE} \rt) \rt) + 2mK^2 \na_a \hat f\\
&\quad + K \pl_a g^A \lt( \na_B W^{(-2)}_A - \na_A W^{(-2)}_B \rt)g^{(-1)B} - K^2 \pl_a g^A \pl_{\bar u} W^{(-2)}_A |\bna\hat f|^2\\
&\quad - \frac{1}{2}K \pl_{\bar u}C^B_D C_{AB} \pl_a g^A g^{(-1)D} + \frac{1}{3} \pl_{\bar u} \pl_{\bar u} C_{AB} g^{(-1)A}g^{(-1)B} \na_a \hat f \\
&\quad  - 2K \pl_{\bar u} W^{(-2)}_A g^{(-1)A} \na_a \hat f + \frac{1}{6} K \pl_{\bar u} \pl_{\bar u} C_{AB} \pl_a g^A g^{(-1)B} |\bna \hat f|^2
\end{align*}
\end{prop}
\begin{proof}
The metric component $g_{ua}^{(-1)}$ is equal to 
\begin{align*}
& 2m K^2 \na_a \hat f - \pl_u a^{(-1)} \na_a \hat f -K \na_a a^{(-1)}\\
& -\na_a a^{(-1)} \pl_u \rho^{(0)} - \na_a \hat f \pl_u \rho^{(-1)} - K \na_a \rho^{(-1)} - \pl_u a^{(-1)} \na_a \rho^{(0)} + \underline{ K^{-2} \na_a K \pl_u a^{(-2)} } + K U^{(-2)} \na_a K\\
& + K^{-2} \sigma_{AB} \pl_a g^A \Big(  \pl_u g^{(-3)B} + K^4  W^{(-3)B} + K^2 \pl_u a^{(-1)} W^{(-2)B}\\
&\hspace{4cm}  - 2K^4 \rho^{(0)} W^{(-2)B} +K^3 \pl_D W^{(-2)B}g^{(-1)D} + K^3 \pl_{\bar u}W^{(-2)B} a^{(-1)}\Big) \\
& + K^{-2} \sigma_{AB} \pl_a g^{(-1)A} \lt( \pl_u g^{(-2)B} + K^3  W^{(-2)B}  \rt)\\
& + K^{-2} \sigma_{AB} \lt( \underline{ \pl_a g^{(-2)A} } + K^2 W^{(-2)A} \na_a \hat f \rt)\pl_u g^{(-1)B}\\
& + \lt( 2K^{-1} \rho^{(0)}\sigma_{AB} + K^{-1}C_{AB} + K^{-2} \pl_D \sigma_{AB} g^{(-1)D} \rt) \cdot \\
& \hspace{4cm}\lt( \underline{ \pl_a g^A \pl_u g^{(-2)B} }+ K^3 \pl_a g^A W^{(-2)B} + \pl_a g^{(-1)A} \pl_u g^{(-1)B} \rt)\\
& + \pl_a g^A \pl_u g^{(-1)B} \lt[ \lt( 2K^{-1}\rho^{(-1)} + \rho^{(0)2} \rt)\sigma_{AB} + 2K^{-1} \rho^{(0)} \pl_D \sigma_{AB} g^{(-1)D}  \rt]\\
& + \pl_a g^A \pl_u g^{(-1)B} \lt[ K^{-2} \lt( \pl_D \sigma_{AB} g^{(-2)D} + \frac{1}{2} \pl_D \pl_E \sigma_{AB} g^{(-1)D}g^{(-1)E} \rt) + \underline{ \rho^{(0)}C_{AB} } \rt]\\
& + \pl_a g^A \pl_u g^{(-1)B} \lt[ K^{-1} \lt( \pl_D C_{AB} g^{(-1)D} + \pl_{\bar u}C_{AB} a^{(-1)}\rt) + \underline{ \frac{1}{4} C_{DE}C^{DE} \sigma_{AB} } \rt].
\end{align*}
We work in a normal coordinate $\bar x^A$ centered at $g(x_0)$. By Lemma \ref{g(-2)} and Lemma \ref{a-2} we have cancellation among underlined terms in the previous expression and get
\begin{align*}
g_{ua}^{(-1)} &= K^2 \pl_a g^A \bar g_{\bar u A}^{(-1)} + 2mK^2 \na_a \hat f - \na_a \hat f \pl_u \rho^{(-1)} - \na_a (K\rho^{(-1)})\\
&\quad  + \rho^{(-1)} \na_a K + \pl_a g^A \pl_u g^{(-1)B} \cdot 2K^{-1}\rho^{(-1)} \sigma_{AB}+ K^{-2}\sigma_{AB}\pl_a g^A \pl_u g^{(-3)B} \\
&\quad - \pl_u a^{(-1)} \na_a \hat f -K \na_a a^{(-1)} \\
&\quad -\na_a a^{(-1)} \pl_u \rho^{(0)} - \pl_u a^{(-1)} \na_a \rho^{(0)} + \frac{1}{4}K^{-1}\na_a K \pl_{\bar u}C_{AB} g^{(-1)A}g^{(-1)B} + K U^{(-2)} \na_a K\\
&\quad + \sigma_{AB} \pl_a g^A \Big( \pl_u a^{(-1)} W^{(-2)B} +K \na_D W^{(-2)B}g^{(-1)D} + K \pl_{\bar u}W^{(-2)B} a^{(-1)}\Big) \\
&\quad + K^{-2} \sigma_{AB} \pl_a g^{(-1)A} \lt( -\frac{1}{2}K^2 \pl_{\bar u}C^B_D g^{(-1)D} \rt) + K\pl_a g^{(-1)A} W_A^{(-2)}\\
&\quad + K^{-2} \sigma_{AB} \lt( -\pl_a \lt( K\rho^{(0)} g^{(-1)B} \rt) -\frac{1}{2} \pl_a \Gamma^B_{ED}g^{(-1)E}g^{(-1)E} - \frac{1}{2}K \pl_a C^A_D g^{(-1)D}  \rt)\pl_u g^{(-1)B} \\
&\quad + W_B^{(-2)} \pl_u g^{(-1)B} \na_a \hat f\\
&\quad + 2\rho^{(0)}\sigma_{AB} \pl_a g^A \lt( -K \pl_u(\rho^{(0)}g^{(-1)B}) - \frac{1}{2} \pl_u C^B_D g^{(-1)D} + K^{-1}\pl_a g^{(-1)A} \pl_u g^{(-1)B}\rt)\\
&\quad + C_{AB} \pl_a g^A \lt( -\pl_u(\rho^{(0)}g^{(-1)B} ) -\frac{1}{2} \pl_u C^B_D g^{(-1)D} \rt) \\
&\quad + \pl_a g^A \pl_u g^{(-1)B} \sigma_{AB}\rho^{(0)2} \\
&\quad + \pl_a g^A \pl_u g^{(-1)B} \cdot \frac{1}{2}K^{-2} \pl_D \pl_E \sigma_{AB} g^{(-1)D}g^{(-1)E} \\
&\quad + \pl_a g^A \pl_u g^{(-1)B} \cdot K^{-1} \lt( \na_D C_{AB} g^{(-1)D} + \pl_{\bar u}C_{AB} a^{(-1)}\rt) ,
\end{align*}
where we collect $\rho^{(-1)}$ and $g^{(-3)B}$ terms in the first two lines. 

We write $g_{ua}^{(-1)} = K^2 \pl_a g^A \bar g_{\bar uA}^{(-1)} + 2mK^2 \na_a \hat f + \RN{1}+\RN{2}$ where $\RN{1}$ collects terms involving the shear tensor $C_{AB}$. By Lemma \ref{rho-1} and Lemma \ref{g-3}, $\RN{1}$  contains
\begin{align*}
&\frac{1}{4}K^{-1}\na_a K \pl_{\bar u}C_{AB} g^{(-1)A} g^{(-1)B} + K U^{(-2)}\na_a K\\
&-K \na_b K \bna^b \hat f pl_a g^A W_A^{(-2)} + K\pl_a g^A \na_D W_A^{(-2)} g^{(-1)D} - \frac{1}{2}K^2 \pl_a g^A \pl_{\bar u} W_A^{(-2)} |\bna \hat f|^2\\
&\uwave{-\frac{1}{2} \pl_{\bar u}C_{AB} \pl_a g^{(-1)A} g^{(-1)B}} + \uwave{K \pl_a g^{(-1)A} W_A^{(-2)}}\\
&-\frac{1}{2}K^{-1} \pl_a C_{BD} g^{(-1)D} \pl_u g^{(-1)B} + \uwave{W_B^{(-2)} \pl_u g^{(-1)B} \na_a \hat f}\\
& \underline{- \rho^{(0)} \pl_a g^A \pl_u C_{AD} g^{(-1)D} - C_{AB} \pl_a g^A \pl_u (\rho^{(0)}g^{(-1)B}) } -\frac{1}{2} K \pl_{\bar u} C^B_D C_{AB} g^{(-1)D}\pl_a g^A\\
& + K^{-1} \na_D C_{AB} g^{(-1)D} \pl_a g^A \pl_u g^{(-1)B} - \frac{1}{2} \pl_{\bar u}C_{AB} \pl_a g^A \pl_u g^{(-1)B} |\bna \hat f|^2 
\end{align*}
plus terms in $- \na_a \hat f \pl_u \rho^{(-1)} - \na_a (K\rho^{(-1)}) + \rho^{(-1)} \na_a K + K^{-2}\sigma_{AB}\pl_a g^A \pl_u g^{(-3)B}$:
\begin{align*}
&\dashuline{ \frac{1}{4}\pl_a g^A \pl_{\bar u} \na_A C_{BD}g^{(-1)B}g^{(-1)D} } + \frac{1}{4}\pl_{\bar u}\pl_{\bar u}C_{BD}\na_a \hat f g^{(-1)B}g^{(-1)D} + \uwave{ \frac{1}{2}\pl_{\bar u}C_{AB}\pl_a g^{(-1)A}g^{(-1)B} }\\
&-\na_a K W_A^{(-2)}g^{(-1)A} - K\pl_a g^A\na_A W^{(-2)}_B g^{(-1)B} - K\pl_{\bar u}W^{(-2)}_B g^{(-1)B}\na_a \hat f\\
& \uwave{ -K W^{(-2)}_A \pl_a g^{(-1)A} } + \pl_a \lt( \frac{1}{4}K^2 \mathcal{C}_{de} \bna^d\bna^e \hat f\rt)\\
&+\frac{1}{12}\pl_{\bar u}\pl_{\bar u}C_{AB}g^{(-1)A}g^{(-1)B}\na_a \hat f + \frac{1}{6}\pl_{\bar u}C_{AB} g^{(-1)A} \pl_ug^{(-1)B}\na_a \hat f + \frac{1}{12}K^{-1}\na_a K \pl_{\bar u}C_{AB}g^{(-1)A}g^{(-1)B}\\
&-\frac{1}{16}K^2\pl_{\bar u}(C_{DE}C^{DE})\na_a \hat f - \frac{1}{16}KC_{DE}C^{DE}\na_a K\\
& \dashuline{ -\frac{1}{4}\pl_{\bar u}\na_A C_{BD}\pl_a g^A g^{(-1)B}g^{(-1)D} } - \frac{1}{2}K^{-1}\na_A C_{BD}\pl_a g^A \pl_u g^{(-1)B}g^{(-1)D}\\
&-\frac{1}{2}K^2 \pl_a g^A \pl_{\bar u}W^{(-2)}_A |\bna\hat f|^2 - K \pl_b K \bna^b \hat f \pl_a g^A W^{(-2)}_A\\
&-K\pl_{\bar u}W^{(-2)}_A g^{(-1)A}\pl_a \hat f - \uwave{ W^{(-2)}_A \pl_u g^{(-1)A} \na_a \hat f } - W^{(-2)}_A g^{(-1)A}\na_a K + \underline{ \pl_u \lt( \rho^{(0)}C_{AB}\pl_a g^A g^{(-1)B} \rt) }\\
&+ \frac{1}{6}K \pl_{\bar u}\pl_{\bar u}C_{AB} \pl_a g^A g^{(-1)B} |\bna\hat f|^2 + \frac{1}{6} \pl_a g^A \pl_u g^{(-1)B} \pl_{\bar u}C_{AB} |\bna \hat f|^2 + \frac{1}{3} \na_b K \bna^b \hat f \pl_a g^A \pl_{\bar u}C_{AB} g^{(-1)B}
\end{align*}
where we underlined cancelled terms.

Next, $\RN{2}$ contains, by Lemma \ref{g(-2)},\begin{align*}
&-\pl_u a^{(-1)}\na_a \hat f - K \na_a a^{(-1)} - \na_a a^{(-1)} \pl_u \rho^{(0)} - \pl_u a^{(-1)} \na_a \rho^{(0)}\\
&+K^{-2}\na_a K \lt( \frac{1}{2} \pl_u \rho^{(0)}\sigma_{AB}g^{(-1)A}g^{(-1)B} + \dashuline{ \rho^{(0)}\sigma_{AB}g^{(-1)A}\pl_u g^{(-1)B} } \rt)\\
&+K^{-2}\sigma_{AB}\pl_a g^{(-1)A}  \lt( -K\pl_u\rho^{(0)} g^{(-1)B} - \dashuline{ K\rho^{(0)}\pl_ug^{(-1)B} } \rt)\\
&+K^{-2}\sigma_{AB} \lt( \dashuline{ -\na_a K \rho^{(0)} g^{(-1)A}} - K\na_a \rho^{(0)}g^{(-1)A}\rt) \pl_u g^{(-1)B}\\
&+K^{-2}\sigma_{AB}\lt( \dashuline{ -K\rho^{(0)}\pl_a g^{(-1)A} } - \frac{1}{2}\pl_a\Gamma^A_{ED}g^{(-1)E}g^{(-1)D} \rt)\pl_ug^{(-1)B}\\
&+2K^{-1}\rho^{(0)}\sigma_{AB}\pl_a g^A \lt( \underline{ -K\pl_u \rho^{(0)}g^{(-1)B} } - \uwave{ K\rho^{(0)}\pl_u g^{(-1)B} } + \dashuline{ \pl_a g^{(-1)A} \pl_u g^{(-1)B} } \rt)\\
&+\pl_a g^A \pl_u g^{(-1)B} \lt( \uwave{ \rho^{(2)2}\sigma_{AB} } + \frac{1}{2}K^{-2}\pl_D\pl_E \sigma_{AB} g^{(-1)D} g^{(-1)E} \rt)
\end{align*}
plus terms from $- \na_a \hat f \pl_u \rho^{(-1)} - \na_a (K\rho^{(-1)}) + \rho^{(-1)} \na_a K + K^{-2}\sigma_{AB}\pl_a g^A \pl_u g^{(-3)B}$:
\begin{align*}
&\na_a \lt( Ka^{(-1)} \rt) - \na_a \lt( \frac{1}{4}K^2 \lt( \bna_b\bna_c \hat f - \frac{1}{2}\bar\Delta \hat f \sigma_{bc} \rt)\lt( \bna^b\bna^c \hat f - \frac{1}{2}\bar\Delta \hat f \bar\sigma^{bc}\rt)\rt)\\
&+\pl_u \lt( \frac{1}{3}a^{(-1)}\na_a \hat f \rt) + \underline{ 2\pl_u \rho^{(0)} \rho^{(0)}\sigma_{AB}\pl_a g^A g^{(-1)B} } + \uwave{ \rho^{(0)2}\sigma_{AB}\pl_a g^A \pl_u g^{(-1)B} } \\
&-K^{-2} \lt( \frac{1}{6}\pl_D \Gamma^F_{EB}\sigma_{AF}  g^{(-1)B} \pl_u g^{(-1)D} \pl_a g^A + \frac{1}{3} \pl_D \Gamma^F_{EB}\sigma_{AF} \pl_u g^{(-1)B} g^{(-1)D} \rt) g^{(-1)E}  \pl_a g^A
\end{align*}
where we underlined cancelled terms. Hence,
\begin{align*}
\RN{2} &= -\frac{1}{4} \na_a \lt( K^2 \lt| \bna\bna \hat f - \frac{1}{2} \bar\Delta \hat f \bar\sigma\rt|^2 \rt) \\
&\quad + \pl_u \rho^{(0)} \lt( -\na_a a^{(-1)} + \frac{1}{2} |\bna\hat f|^2 \na_a K - \frac{1}{2}K^{-1} \na_a \lt( K^2 |\bna\hat f|^2 \rt) \rt) \\
&\quad + \na_a \rho^{(0)} \lt( -\pl_u a^{(-1)} - K^{-1} \sigma_{AB} g^{(-1)A} \pl_u g^{(-1)B} \rt)\\
&\quad - \frac{2}{3} \pl_u a^{(-1)} \na_a \hat f + \frac{4}{3} a^{(-1)} \na_a K \\
&\quad + K^{-2} \lt[ -\frac{1}{2}\sigma_{FB}\pl_A \Gamma^F_{ED} + \frac{1}{2} \pl_D\pl_E \sigma_{AB} - \frac{1}{6} \pl_B \Gamma^F_{ED} \sigma_{AF} - \frac{1}{3} \pl_D\Gamma^F_{EB}\sigma_{AF} \rt] \cdot \\
&\hspace{10cm} g^{(-1)D} g^{(-1)E} \pl_a g^A \pl_u g^{(-1)B}.
\end{align*}
It is easy to see both the second and the third lines vanish. We claim that the last two lines cancel. Indeed, we have
$\pl_D \pl_E \sigma_{AB} = \pl_D \lt( \Gamma^F_{EA} \sigma_{FB} + \Gamma^F_{EB} \sigma_{AF} \rt) $ and the bracket in the last line becomes
\begin{align*}
\frac{1}{2} \pl_D \Gamma^F_{EA}\sigma_{FB} - \frac{1}{2} \pl_A \Gamma^F_{ED}\sigma_{FB} + \frac{1}{6} \pl_D \Gamma^F_{EB}\sigma_{AF} - \frac{1}{6}\pl_B \Gamma^F_{ED}\sigma_{AF} = \frac{1}{2} R_{DABE} + \frac{1}{6}R_{DBAE}\\
= \frac{1}{2} \lt( \sigma_{DB}\sigma_{AE} - \sigma_{DE}\sigma_{AB}\rt) + \frac{1}{6}\lt( \sigma_{DA}\sigma_{BE} - \sigma_{DE} \sigma_{BA} \rt).
\end{align*}

Moreover, we observe that the terms
\begin{align*}
K U^{(-2)}\na_a K + \frac{1}{4} \na_a \lt( K^2 \mathcal{C}_{de} \bna^d \bna^e \hat f \rt) - \frac{1}{16} K^2 \pl_{\bar u}(C_{DE}C^{DE}) \na_a \hat f - \frac{1}{16} K C_{DE}C^{DE} \na_a K\\
- \frac{1}{4} \na_a \lt( K^2 \lt| \bna\bna \hat f - \frac{1}{2} \bar\Delta \hat f \bar\sigma \rt|^2 \rt)
\end{align*}
can be rearranged as
\[ - \frac{1}{16} \na_a \lt( \sigma^{bd}\sigma^{ce} g^{(1)}_{bc} g^{(1)}_{de} \rt) + \frac{1}{16}K^2 \pl_a \na_A (C_{DE}C^{DE}). \]
 
Putting these together, we obtain
\begin{align*}
&g_{ua}^{(-1)} + \frac{1}{16} \na_a \lt( \sigma^{bd}\sigma^{ce} g^{(1)}_{bc} g^{(1)}_{de} \rt)\\
&= K^2 \pl_a g^A \lt( \bar g^{(-1)}_{\bar u A} + \frac{1}{16}\na_A \lt( C_{DE}C^{DE} \rt) \rt) + 2mK^2 \na_a \hat f\\
&\quad + K \pl_a g^A \lt( \na_B W^{(-2)}_A - \na_A W^{(-2)}_B \rt)g^{(-1)B} - K^2 \pl_a g^A \pl_{\bar u} W^{(-2)}_A |\bna\hat f|^2\\
&\quad - \frac{1}{2}K \pl_{\bar u}C^B_D C_{AB} \pl_a g^A g^{(-1)D} + \frac{1}{3} \pl_{\bar u} \pl_{\bar u} C_{AB} g^{(-1)A}g^{(-1)B} \na_a \hat f \\
&\quad  - 2K \pl_{\bar u} W^{(-2)}_A g^{(-1)A} \na_a \hat f + \frac{1}{6} K \pl_{\bar u} \pl_{\bar u} C_{AB} \pl_a g^A g^{(-1)B} |\bna \hat f|^2\\
&\quad + \frac{1}{3}K^{-1}\na_a K \pl_{\bar u}C_{AB}g^{(-1)A}g^{(-1)B} - 2K \na_b K \bna^b \hat f \pl_a g^A W_A^{(-2)}\\
&\quad + \pl_a g^A \pl_u g^{(-1)B} \lt( K^{-1}\na_D C_{AB}g^{(-1)D} - K^{-1}\na_A C_{BD} g^{(-1)D} - \frac{1}{3}\pl_{\bar u}C_{AB} |\bna \hat f|^2 \rt)\\
&\quad - \frac{1}{3}K^{-1} \pl_{\bar u}C_{AB} g^{(-1)A}\pl_u g^{(-1)B} \na_a \hat f - 2 \na_a K W_A^{(-2)} g^{(-1)A} + \frac{1}{3}\na_b K \bna^b \hat f \pl_a g^A \pl_{\bar u}C_{AB} g^{(-1)B}.
\end{align*}
We claim that the last three lines vanish. Indeed, it is straightforward to check the terms involving 
$\pl_{\bar u}C_{AB}$ vanish by Lemma \ref{g(-1)}. We are left with
\begin{align*}
-K \na_b K \bna^b \hat f \pl_a g^A \na^B C_{AB} -  \pl_a g^A \pl_b g^B \bna^b K (\na_D C_{AB} - \na_A C_{BD}) g^{(-1)D} - \na_a K \na^B C_{AB} g^{(-1)A},
\end{align*}
which vanishes by Proposition A.1 of \cite{KWY}, $\na_D C_{AB} - \na_A C_{DB} = \na^E C_{AE} \sigma_{DB} - \na^E C_{DE} \sigma_{AB}$. This completes the proof.    
\end{proof}

\section{Dray-Streubel angular momentum of a general section}
As an application of the transformation formula, we study the Dray-Streubel angular momentum \cite{DS} of a general section. Consider a $u=$constant section, say $u=0$, in $\mathscr{I}^+$ that corresponds to $r = \infty$ in a Bondi-Sachs coordinate system $(u,r,x^A)$. Its Dray-Streubel angular momentum  is defined as
\begin{align*}
J(u=0) = \int_{S^2} Y^A \lt( N_A(0,x) - \frac{1}{4}\lt( C_{AB}\na_D C^{BD}\rt)(0,x) \rt).
\end{align*}
Here $Y^A$  is a rotation Killing vector field on $S^2$ and the integral $\int_{S^2}$ is taken with respect to the standard metric $\sigma_{AB}$ (we omit the area form).

Let $f(x)$ be a positive function on $S^2$. The Dray-Streubel angular momentum of the section $u=f(x)$ is defined as the angular momentum of the section $u=0$ plus the Ashtekar-Streubel flux \cite{AS}
\begin{align} \label{DS1}
J(u=f) = J(u=0) - \frac{1}{4} \int_{\Omega} N^{AB} \mathcal{L}_Y C_{AB}
\end{align}
where $\mathcal{L}$ denotes the Lie derivative. Here $\Omega$ is the region in $\mathcal{I}^+$ bounded by the $u=0$ and $u=f(x)$ sections. The integral of a function $F(u,x)$ on $\mathcal{I}^+$ is defined as $\int_\Omega F(u,x) = \int_{S^2} \lt( \int_{0}^{f(x)} F(u',x) du' \rt)$.

There is another definition of $J(u=f)$, see \cite{Geroch, CPWWWY}.  Consider the supertranslation $u = u' + f(x)$ under which the $u=f(x)$ section corresponds to the $u'=0$ section. The Dray-Streubel angular momentum of $u=f(x)$ can be defined alternatively as
\begin{align}\label{DS2}
J(u=f) = \int_{S^2} Y^A \lt( N_A' - \frac{1}{4} C'_{AB} \na_D C^{'BD} \rt) - 2 Y^A \na_A f m'
\end{align}
where $N_A', C_{AB}', m'$ are the Bondi data of $u'=0$ hypersurface in the Bondi-Sachs coordinate system $(u',r',x^A)$. The right-hand side is physically the charge of the $u'=0$ section associated with the BMS vector field $Y^A \frac{\pl}{\pl X^A} - Y^A \na_A f \frac{\pl}{\pl u'}$.

It is known that the two definitions are equivalent. In the first subsection, we give a direct proof of this fact using the transformation formulae for the Bondi data. In the second subsection, we reformulate our calculations in terms of  differential forms and obtain an explicit formula for the Drey-Streubel angular momentum of a general section. In the third subsection, the results are generalize to the BMS charge of a general section.

\subsection{The equivalence of the two definitions}

The goal of this subsection is to prove the following theorem using the transformation formulae for the Bondi data, asserting the equivalence of \eqref{DS1} and \eqref{DS2}:
\begin{theorem}\label{DS}
\begin{align*}
&\int_{S^2} Y^A \lt( N_A' - \frac{1}{4} C'_{AB} \na_D C^{'BD} - 2 Y^A \na_A f m' \rt)\\
&= \int_{S^2} Y^A \lt( N_A (0,x) - \frac{1}{4} \lt( C_{AB}\na^D C_{BD} \rt)(0,x) \rt) - \frac{1}{4}\int_\Omega N^{AB} \mathcal{L}_Y C_{AB}
\end{align*}
\end{theorem}
The proof relies on two lemmas. To simply the presentation, in their proofs we suppress the dependence of $x\in S^2$ in the second slot if there is no confusion. For example, $C_{AB}(f)$ stands for $C_{AB}(f(x),x)$. Moreover, we will make use of the following identities for the rotation Killing vector $Y^A:$
\[ \na_A Y^A =0, \na^B Y^A = -\na^A Y^B, \Delta Y^A = - Y^A.\]

\begin{lem}
\begin{align}\label{evolution}
\pl_u \lt( Y^A \lt( N_A - \frac{1}{4}C_{AB}\na_D C^{BD} \rt) \rt) - \na_D V^D = -\frac{1}{4}N^{AB} \mathcal{L}_Y C_{AB} 
\end{align}
where the divergence of the vector field $V$ on $\mathscr{I}^+$ takes the form
\begin{align*}
\na_D V^D = \na_A \lt( Y^A m\rt) + \frac{1}{2}\na_D \lt(Y^A C_{AB}N^{BD} \rt)  - \frac{1}{4} Y^A \na^B P_{BA} 
\end{align*}
with $P_{BA} = \na_B\na^D C_{DA} - \na_A\na^D C_{DB}$.
\end{lem}
\begin{proof}
By the evolution formula of $N_A$ (see \cite[(5.103)]{CJK} for example) 
\begin{align*}
 \partial_u N_A &= \na_A m -\frac{1}{4}\nabla^B P_{BA} +\frac{1}{4}\nabla_A(C_{BE} N^{BE})-\frac{1}{4}\nabla_B (C^{BD} N_{DA})+\frac{1}{2} C_{AB}\nabla_D N^{DB},
\end{align*}
we have
\begin{align*}
&\pl_u \lt( Y^A \lt( N_A - \frac{1}{4}C_{AB}\na_D C^{BD} \rt) \rt) \\
&= Y^A \lt[ \na_A m -\frac{1}{4}\nabla^B P_{BA} +\frac{1}{4}\nabla_A(C_{BE} N^{BE})\rt]\\
&+Y^A\lt[-\frac{1}{4}\nabla_B (C^{BD} N_{DA})+\frac{1}{4} C_{AB}\nabla_D N^{DB} - \frac{1}{4}N_{AB} \na_D C^{BD}\rt].
\end{align*}

The last line can be rewritten as 
\[\begin{split} &-\frac{1}{4} \na_B \lt( Y^A C^{BD}N_{DA} \rt) + \frac{1}{4}\na_B Y^A C^{BD}N_{DA} + \frac{1}{4} \na_D \lt( Y^A C_{AB}N^{DB}\rt) - \frac{1}{4}\na_D Y^A C_{AB} N^{DB} \\
&\quad -\frac{1}{4} Y^A \na_D C_{AB}  N^{DB} - \frac{1}{4}N_{AB} Y^A \na_D C^{BD}\end{split}\]

Rearranging indexes and applying Lemma \ref{symmetric_traceless}, we get 
\begin{equation}
\begin{split}
& \frac{1}{4}\na_B Y^A C^{BD}N_{DA}  - \frac{1}{4}\na_D Y^A C_{AB} N^{BD} -\frac{1}{4} Y^A \na_D C_{AB}  N^{BD} - \frac{1}{4}N_{AB} Y^A \na_D C^{BD}\label{Lie-d}\\
&= \frac{1}{4}N_{AB} \na_D Y^A C^{BD} - \frac{1}{4} N_{AB} \na^A Y^D C_{DB} - \frac{1}{4} Y^A \na_A C_{BD} N^{BD}. 
\end{split}\end{equation}

Recalling the definition of Lie derivative
\[
 \mathcal{L}_Y C_{AB} = Y^D \na_D C_{AB} + \na_A Y^D C_{DB} + \na_B Y^D C_{AD},
\] we recognize that the last line of \eqref{Lie-d} is simply $-\frac{1}{4} N^{AB}\mathcal{L}_Y C_{AB}$. 

Hence
\begin{align*}
&\pl_u \lt( Y^A \lt( N_A - \frac{1}{4}C_{AB}\na_D C^{BD} \rt) \rt) \\
&= \na_A (Y^A m) - \frac{1}{4} Y^A \na^B P_{BA}  -\frac{1}{4}N^{AB}\mathcal{L}_Y C_{AB}\\
&\quad + \frac{1}{4} \na_A \lt( Y^A C_{BD}N^{BD}\rt) - \frac{1}{4} \na_A \lt( Y^B C^{AD}N_{DB} \rt) + \frac{1}{4} \na_A \lt( Y^B C_{BD} N^{DA} \rt)
\end{align*}
and by \eqref{product of symmetric traceless 2-tensors} the last line is equal to $\frac{1}{2}\na_A \lt( Y^B C_{BD} N^{DA}\rt)$. \end{proof}

$V^D$ can be found explicitly as 
\begin{equation} \label{V-formula} V^D=Y^D m+\frac{1}{2} Y^AC_{AB} N^{BD}-\frac{1}{4} Y^AP^{D}_{\,\,\,\,A}+\frac{1}{2} \na^D Y^A \na^B C_{BA}+\frac{1}{2}Y^AC^D_{\,\,\,\,A}\end{equation} where we use
\begin{align*}
\frac{1}{4}\na^B Y^A P_{BA} = \frac{1}{2}\na^B Y^A \na_B \na^D C_{DA} = \frac{1}{2} \na_B \lt( \na^B Y^A \na^D C_{DA}\rt) + \frac{1}{2} Y^A \na^D C_{DA}
\end{align*}

\begin{lem}\label{using transformation formula}
\begin{align*}
&\int_{S^2} Y^A \lt( N_A' - 2m' \na_A f - \frac{1}{4}C'_{AB}\na_D C^{'DB} \rt) = \int_{S^2} Y^A \lt( N_A - \frac{1}{4}C_{AB}\na_D C^{DB}\rt)+ \int_{S^2}  V^D\na_D f,\\
\end{align*} where $V^D$ is given in \eqref{V-formula}. 
On the right-hand side the Bondi data $N_A, m, C_{AB}, N_{AB}, P_{BA}$ are evaluated at $(f(x),x)$ and then integrated over $S^2$.
\end{lem}
\begin{proof}

By the transformation formula of angular momentum aspect \eqref{transformation_angular_general}, we have
\begin{align*}
N_A'(x) &= N_A + 3m\na_A f - \frac{3}{4}P_{BA}\na^B f - \frac{3}{4}\na^B N_{AB}|\na f|^2 + \frac{3}{4}N^B_D C_{AB}\na^D f \\
&\quad + \frac{1}{2} \pl_u N_{DB} \na^D f \na^B f \na_A f + \frac{3}{2} \na^B N_{DB} \na^D f\na_A f - \frac{1}{4} \pl_u N_{AB}\na^B f |\na f|^2
\end{align*}
where the Bondi data on the right-hand side are evaluated at $(f(x),x)$.
We rewrite three terms on the right-hand side of the above equation using the following:
\begin{align*}
-\frac{3}{4}\na^B N_{AB} - \frac{1}{4}\pl_u N_{AB}\na^B f = -\frac{1}{2}\na^B N_{AB} - \frac{1}{4} \na^B \lt( N_{AB}(f) \rt)
\end{align*}
and 
\begin{align*}
\begin{split}\frac{1}{2}\pl_u N_{DB}\na^D f \na^B f
&= \frac{1}{2}\na^B\na^D \lt(C_{BD}(f) \rt) - \frac{1}{2}(\na^B\na^D C_{BD})(f) - \na^D N_{BD}(f)\na^B f \\
&\quad - \frac{1}{2}N_{BD}(f)\na^B\na^D f.\end{split}
\end{align*}

By the transformation formula of mass aspect \eqref{transformation_mass_general}, we have
\begin{align*}
m'(x) = m(f) + \frac{1}{4}\na^B \na^D \lt( C_{BD}(f)\rt) - \frac{1}{4} (\na^B\na^D C_{BD})(f).
\end{align*}
Consequently, we get
\begin{equation}\label{N-A1}
\begin{split}N_A'(x) - 2 m'(x)\na_A f &=N_A + m\na_A f - \frac{3}{4}P_{BA}\na^B f - \frac{1}{2}\na^B N_{AB}|\na f|^2 - \frac{1}{4}\na^B ( N_{AB}(f)) |\na f|^2\\
&\quad + \frac{3}{4}N^B_D C_{AB}\na^D f + \frac{1}{2}\na^D N_{BD}\na^B f\na_A f - \frac{1}{2}N_{BD}\na^B\na^D f \na_A f\end{split}
\end{equation}
where the Bondi data on the right-hand side are evaluated at $(f(x),x)$.

Recalling $P_{BA} = \na_B\na^D C_{DA} - \na_A\na^D C_{DB}$, at $(f(x), x)$ we rearrange
\begin{align}\begin{split}\label{-3/4P}
&-\frac{3}{4}P_{BA}\\ &= -\frac{1}{4} P_{BA}  -\frac{1}{2} \na_B\lt ( \na^D C_{DA}(f) \rt) + \frac{1}{2}\na^D N_{DA} \na_B f + \frac{1}{2}\na_A \lt( \na^D C_{DB}(f) \rt) - \frac{1}{2}\na^D N_{DB}\na_A f. \end{split}\end{align}

Plugging this into \eqref{N-A1}, we obtain
\begin{align*}
N_A' - 2m' \na_A f &= N_A + m \na_A f - \frac{1}{4} P_{BA}\na^B f  \\
&\quad - \frac{1}{2}\na_B \lt( \na^D C_{DA}(f)\na^B f \rt)    + \frac{1}{2}\na^D C_{DA}\Delta f \\
&\quad + \frac{1}{2} \na_A \lt( \na^D C_{DB}(f) \na^B f \rt) -\frac{1}{2}\na^D C_{DB}\na_A\na^B f\\
&\quad -\frac{1}{4}\na^B \lt( N_{AB}(f)|\na f|^2 \rt) + \frac{1}{2}N_{AB}\na^B\na_D f \na^D f\\
&\quad + \frac{3}{4} N^B_D C_{AB}\na^D f - \frac{1}{2}N_{BD}\na^B\na^D f \na_A f,
\end{align*}
where the second line on the right hand side comes from $-\frac{1}{2} \na_B\lt ( \na^D C_{DA}(f) \rt)$ on the right hand side of \eqref{-3/4P}
, the third line on the right hand side comes from $ \frac{1}{2}\na_A \lt( \na^D C_{DB}(f) \rt) $ on the right hand side of \eqref{-3/4P}, and the fourth line comes from $- \frac{1}{4}\na^B \lt[ N_{AB}(f(x),x)\rt] |\na f|^2$ on the right hand side of \eqref{N-A1}.
Note that both terms $\frac{1}{2} \na_A \lt( \na^D C_{DB}(f) \na^B f \rt)$ and  $-\frac{1}{4}\na^B \lt( N_{AB}(f)|\na f|^2 \rt)$ vanish when integrating against $Y^A$. 
On the other hand,
\begin{align*}
-\frac{1}{4}C'_{AB}\na_D C^{'DB} &= -\frac{1}{4}C_{AB}\na_D C^{BD} - \frac{1}{4}C_{AB}N^{BD}\na_D f + \frac{1}{4}C_{AB}\na^B (\Delta + 2 )f\\
&\quad + \frac{1}{2}\na_D C^{BD}F_{AB} + \frac{1}{2}N^{BD}F_{AB}\na_D f\\
&\quad -\frac{1}{2} F_{AB} \na^B(\Delta + 2)f,
\end{align*} where $F_{AB} = \na_A \na_B f - \frac{1}{2}\Delta f \sigma_{AB}$.
Note that the last line vanishes when integrating against $Y^A$.

Putting these together, we get
\begin{align*}
&\int_{S^2} Y^A \lt( N_A' - 2m' \na_A f - \frac{1}{4}C'_{AB}\na_D C^{'DB} \rt) \\
&= \int_{S^2} Y^A \lt( N_A + m \na_A f - \frac{1}{4}C_{AB}\na_D C^{DB} -\frac{1}{4}P_{BA}\na^B f +\frac{1}{2}C_{AB}N^{BD}\na_D f- \frac{1}{2}\na_B \lt( \na^D C_{DA}(f)\na^B f \rt)  \rt)\\
&+\int_{S^2} Y^A \Xi_A
\end{align*}  where 
\[\begin{split}\Xi_A=
 &\frac{1}{2}\na^D C_{DA}\Delta f 
-\frac{1}{2}\na^D C_{DB}\na_A\na^B f
 + \frac{1}{4}C_{AB}\na^B (\Delta + 2 )f
 + \frac{1}{2}\na_D C^{BD}F_{AB} \\
 +&\frac{1}{2}N_{AB}\na^B\na_D f \na^D f
 - \frac{1}{2}N_{BD}\na^B\na^D f \na_A f+ \frac{1}{2}N^{BD}F_{AB}\na_D f\end{split}\]

We simplify
\[\frac{1}{2}\na^D C_{DA}\Delta f -\frac{1}{2}\na^D C_{DB}\na_A\na^B f+ \frac{1}{2}\na_D C^{BD}F_{AB}= \frac{1}{4}\na^D C_{DA}\Delta f. \]

Per $\na^D C_{AD} = \na^D (C_{AD}(f)) - N_{AD}\na^D f$, we write
\begin{align*}
\na^D C_{AD} \Delta f = \na^D \lt( C_{AD}(f) \Delta f \rt) - C_{AD} \na^D\Delta f - N_{AD}\na^D f \Delta f.
\end{align*}

Therefore,
\[\begin{split}&\int_{S^2} Y^A (\Xi_A-\frac{1}{2}C_{AB}\na^B f)\\
&=\int_{S^2}Y^A \lt[ - \frac{1}{4} N_{AD} \na^D f \Delta f+ \frac{1}{2}N_{AB}\na^B\na_D f \na^D f
 - \frac{1}{2}N_{BD}\na^B\na^D f \na_A f+ \frac{1}{2}N^{BD}F_{AB}\na_D f\rt]\end{split}\]
 
 One checks that the bracket term in the last line equals \[\frac{1}{2} \lt( N_{AB} F^B_D + N_D^B F_{BA} - N_{BE} F^{BE} \sigma_{AD} \rt) \na^D f,\] which vanishes by virtue of \eqref{product of symmetric traceless 2-tensors}.

Putting these together, we get
\begin{align*}
&\int_{S^2} Y^A \lt( N_A' - 2m' \na_A f - \frac{1}{4}C'_{AB}\na_D C^{'DB} \rt) \\
&= \int_{S^2} Y^A \lt( N_A + m \na_A f - \frac{1}{4}C_{AB}\na_D C^{DB} \rt) + \int_{S^2} -\frac{1}{4}Y^A P_{BA}\na^B f + \frac{1}{2} \na_B Y^A \na^D C_{DA}\na^B f\\
&\quad + \int_{S^2} Y^A \lt( \frac{1}{2}C_{AB}N^{BD}\na_D f + \frac{1}{2}C_{AB}\na^B f \rt).\end{align*}

\end{proof}

\begin{proof}[Proof of Theorem \ref{DS}] Integrating \eqref{evolution} over $\Omega$ and applying Lemma \ref{using transformation formula}  leads to Theorem \ref{DS}.

\end{proof}

\subsection{Interpretation in differential forms}
In this subsection, we reinterpret the flux integral $-\frac{1}{4} \int_{\Omega} N^{AB} \mathcal{L}_Y C_{AB}$ in \eqref{DS1} as the integral of a differential 3-form \eqref{domega} on $\mathscr{I}+$, and identify an exact expression for the Drey-Streubel angular momentum of a general section in terms of a differential 2-form \eqref{2-form}.

We first write \eqref{evolution} in terms of differential forms. Let $\slashed\epsilon = \epsilon_{AB} dx^A dx^B$ denote the area form of $(S^2,\sigma)$ and $\slashed d$ denote the differential of $S^2$.

Define the 2-form $\omega(u,x)$ on $\mathscr{I}^+$
\begin{align}\label{2-form}
\omega = Y^A \lt( N_A - \frac{1}{4}C_{AB}\na_D C^{BD} \rt) \slashed\epsilon - (V \lrcorner \slashed\epsilon) \wedge du.
\end{align} 
Applying the identity
\begin{align*}
\slashed d \lt( V \lrcorner \slashed\epsilon \rt) = (\na_A V^A) \slashed\epsilon
\end{align*}
to time-dependent vector fields on $S^2$, we see that \eqref{evolution} is equivalent to
\begin{align}\label{domega}
d\omega = -\frac{1}{4} N^{AB} \mathcal{L}_Y C_{AB} \slashed\epsilon \wedge du.
\end{align}

By Stokes' Theorem,
\begin{align*}
\int_\Omega - \frac{1}{4} N^{AB} \mathcal{L}_Y C_{AB} = \int_{ u=f(x)} \omega - \int_{S^2} Y^A \lt( N_A(0,x) - \frac{1}{4}(C_{AB}\na_D C^{BD})(0,x) \rt).
\end{align*}
On the section $u = f(x)$, we have
\begin{align*}
\lt( V \lrcorner \slashed{\epsilon} \rt) \wedge du  = \sum_{A < B} \lt( V^D \epsilon_{DA} \na_B f - V^D \epsilon_{DB} \na_A f \rt) dx^A \wedge dx^B = - V^D \na_D f \slashed{\epsilon}
\end{align*} and hence 
\begin{align*}
\int_{u=f(x)} \omega = \int_{S^2} Y^A \lt( N_A - \frac{1}{4}C_{AB}\na_D C^{BD} \rt) + V^D \na_D f
\end{align*}
where $N_A, C_{AB}, V^D$ are evaluated at $(f(x),x)$. We thus obtain the following explicit formula:

\begin{prop} The Drey-Streubel angular momentum of a section $u=f(x)$ is given by the integral of the restriction of the two-form $\omega (u, x)$ \eqref{2-form} to the section $u=f(x)$. 

\end{prop}

\subsection{Generalization to the BMS charge of a general section}
In this subsection, we generalize the results in the last two subsections to the BMS charge of a general section of null infinity. Consider a $u=$ constant section, say $u=u_0$, in $\mathscr{I}^+$ that corresponds to $r = \infty$ in a Bondi-Sachs coordinate system $(u,r,x^A)$. The BMS charge \cite{DS} (see also \cite[(3.5)]{FN}) of $u=u_0$ section for a BMS vector field $\lt( Y^0 + uY^1 \rt) \pl_u + Y^A \pl_A$ is defined as
\begin{align}\label{u=const}
Q(u=u_0) = \int_{S^2} Y^A \lt( N_A - \frac{1}{4} C_{AB}\na_D C^{BD} \rt) + \frac{1}{8}Y^1 C_{BD}C^{BD} + \lt(Y^0 + u_0 Y^1 \rt)\cdot 2m.
\end{align}
where the Bondi data $N_A, C_{AB}, m$ are evaluated at $(u_0, x)$ and the integral $\int_{S^2}$ is taken with respect to the standard metric $\sigma_{AB}$ (we omit the area form). Here $Y^0$ is a function on $S^2$, $Y^A$ is a conformal Killing vector field on $S^2$ with $\na_A Y^A = 2Y^1$.

Let $f(x)$ be a positive function on $S^2$. For a BMS vector field $uY^1 \pl_u + Y^A \pl_A$, there are two equivalent definitions for the BMS charges of the section $u=u_0 + f(x)$. The first is defined as the BMS charge of the section $u=u_0$ plus the Ashtekar-Streubel flux \cite{AS} 
\begin{align}
Q_1(u=u_0+ f) = Q(u=u_0) +  \int_{\Omega} - \frac{1}{4}N^{AB} \mathcal{L}_Y C_{AB} -\frac{1}{4} u Y^1 N_{BD}N^{BD} + \frac{1}{4} Y^1 C_{BD}N^{BD} 
\end{align}
where $\mathcal{L}$ denotes the Lie derivative and $\Omega$ is the region in $\mathscr{I}^+$ bounded by the $u=u_0$ and $u=u_0+f(x)$ sections. 
Consider the supertranslation $u = u' + f(x)$ under which the $u=u_0 + f(x)$ section corresponds to the $u'=u_0$ section. The second definition is given by
\begin{align}
\begin{split}
Q_2(u=u_0 + f) &= \int_{S^2} Y^A \lt( N_A' - \frac{1}{4} C'_{AB} \na_D C^{'BD} \rt) + \frac{1}{8} Y^1 C'_{BD}C^{'BD} \\
&\quad + \int_{S^2} \lt( -  Y^A \na_A f + u_0 Y^1 + f Y^1 \rt)\cdot 2m'
\end{split}
\end{align}
where $N_A', C_{AB}', m'$ are the Bondi data of $u'=u_0$ hypersurface in the Bondi-Sachs coordinate system $(u',r',x^A)$. The right-hand side is physically the BMS charge of the $u'=u_0$ section for the BMS vector field $Y^A \frac{\pl}{\pl X^A} + \lt( fY^1 + u_0 Y^1 - Y^A \na_A f \rt) \frac{\pl}{\pl u'}$.

The equivalence of the two definitions was proved by Shaw \cite{Shaw} and Dray \cite{Dray} using spinor calculus. We give a proof of the equivalence using the transformation formulae for the Bondi data and tensor calculus.

To simply the presentation, we suppress the dependence of $x\in S^2$ in the second slot if there is no confusion. For example, $C_{AB}(f)$ stands for $C_{AB}(u_0 + f(x),x)$. Moreover, we will make use of the following identities for the conformal Killing vector field $Y^A:$
\[ \na^A Y^B + \na^B Y^A = 2Y^1 \sigma^{AB}, \na_A Y^A = 2Y^1, \Delta Y^A = - Y^A, \na_A\na_B Y^1 = -Y^1 \sigma_{AB}.\]

We begin with the evolution formula of the integrand in \eqref{u=const}.
\begin{lem}
\begin{align}\label{evolution}
\begin{split}
&\pl_u \lt( Y^A \lt( N_A - \frac{1}{4}C_{AB}\na_D C^{BD} \rt) + \frac{1}{8} Y^1 C_{BD}C^{BD} + uY^1 \cdot 2m \rt) - \na_D \tilde{V}^D \\
&= -\frac{1}{4}N^{AB} \lt( \mathcal{L}_Y C_{AB} + uY^1 N_{AB} - Y^1 C_{AB} \rt) \end{split} 
\end{align}
where the divergence of the vector field $\tilde{V}$ on $\mathscr{I}^+$ takes the form
\begin{align}
\begin{split}
\na_D \tilde{V}^D &= \na_A \lt( Y^A m\rt) - \frac{1}{4}\na^B (Y^A P_{BA}) + \frac{1}{2}\na_D \lt(Y^A C_{AB}N^{BD} \rt)\\
&\quad +\frac{1}{2} \na_B (\na^B Y^A \na^D C_{DA}) + \frac{1}{2}\na^D (Y^A C_{DA})\\
&\quad - \frac{1}{2} \na_A \lt( Y^1 \na_B C^{BA} \rt) + \frac{1}{2} \na^D \lt( \na^A Y^1 C_{DA} \rt)\\
&\quad + \na_B \lt( \frac{u}{2} Y^1 \na_D N^{BD} \rt) - \na_A \lt( \frac{u}{2} \na_B Y^1 N^{AB} \rt)
\end{split}
\end{align}
with $P_{BA} = \na_B\na^D C_{DA} - \na_A\na^D C_{DB}$.
\end{lem}
\begin{proof}
By the evolution formula of $N_A$ (see \cite[(5.103)]{CJK} for example) 
\begin{align*}
 \partial_u N_A &= \na_A m -\frac{1}{4}\nabla^B P_{BA} +\frac{1}{4}\nabla_A(C_{BE} N^{BE})-\frac{1}{4}\nabla_B (C^{BD} N_{DA})+\frac{1}{2} C_{AB}\nabla_D N^{DB},
\end{align*}
we have
\begin{align}
&\pl_u \lt( Y^A \lt( N_A - \frac{1}{4}C_{AB}\na_D C^{BD} \rt) + \frac{1}{8}Y^1 C_{BD}C^{BD} + uY^1\cdot 2m \rt) \notag\\
&= Y^A \lt[ -\frac{1}{4}\nabla^B P_{BA} +\frac{1}{4}\nabla_A(C_{BE} N^{BE})-\frac{1}{4}\nabla_B (C^{BD} N_{DA})+\frac{1}{4} C_{AB}\nabla_D N^{DB} - \frac{1}{4}N_{AB} \na_D C^{BD} \rt] \label{RHS_line1}\\
&\quad + Y^A \na_A m + \frac{1}{4}Y^1 C_{BD}N^{BD} + Y^1\cdot 2m + uY^1 \lt( -\frac{1}{4} N_{BD}N^{BD} + \frac{1}{2} \na^B\na^D N_{BD} \rt)\label{RHS_line2} 
\end{align}
Denoting the last 4 terms of \eqref{RHS_line1} by $\frac{1}{4}I$, we have 
\begin{align*}
I &= \na_A \lt( Y^A C_{BD} N^{BD} \rt) - 2Y^1 C_{BD}N^{BD} - \na_B \lt( Y^A C^{BD}N_{DA} \rt) + \na_B Y^A C^{BD}N_{DA} \\
&\quad + \na_D \lt( Y^A C_{AB}N^{DB}\rt) - \na^D Y^A C_{AB} N^{BD} -  Y^A \na_D C_{AB}  N^{BD} - N_{AB} Y^A \na_D C^{BD} 
\end{align*}
Applying Lemma \ref{symmetric_traceless} to the second to the last term in the second line, we get
\begin{align*}
I &= \na_A \lt( Y^A C_{BD} N^{BD} \rt) - \na_B \lt( Y^A C^{BD}N_{DA} \rt) + \na_D \lt( Y^A C_{AB}N^{DB}\rt)\\
&\quad + \lt( -2Y^1 C_{AB} \rt)N^{AB} + N_{AB}\na_D Y^A C^{BD} - N_{AB}\na^A Y^D C_{DB} - Y^A \na_A C_{BD}N^{BD}.
\end{align*}
Recalling the definition of Lie derivative
\[
 \mathcal{L}_Y C_{AB} = Y^D \na_D C_{AB} + \na_A Y^D C_{DB} + \na_B Y^D C_{AD},
\] we recognize that the last line is simply $-  \mathcal{L}_Y C_{AB} N^{AB}$. 

Hence \eqref{RHS_line1} is equal to
\begin{align*}
& - \frac{1}{4} \na^B (Y^A P_{BA}) + \frac{1}{4}\na^B Y^A P_{BA} -\frac{1}{4}N^{AB}\mathcal{L}_Y C_{AB}\\
&\quad + \frac{1}{4} \na_A \lt( Y^A C_{BD}N^{BD}\rt) - \frac{1}{4} \na_A \lt( Y^B C^{AD}N_{DB} \rt) + \frac{1}{4} \na_A \lt( Y^B C_{BD} N^{DA} \rt)
\end{align*}
and by \eqref{product of symmetric traceless 2-tensors} the last line is equal to $\frac{1}{2}\na_A \lt( Y^B C_{BD} N^{DA}\rt)$. We further simplify 
\begin{align*}
&\frac{1}{4}\na^B Y^A P_{BA} \\
&= \frac{1}{8} \lt( \na^B Y^A - \na^A Y^B \rt) P_{BA} \\
&= \frac{1}{4}\na^B Y^A \na_B \na^D C_{DA} -\frac{1}{4} \na^A Y^B \na_B\na^D C_{DA} \\
&= \frac{1}{4} \na_B \lt( \na^B Y^A \na^D C_{DA}\rt) + \frac{1}{4} Y^A \na^D C_{DA}- \frac{1}{4} \na_B \lt( \na^A Y^B \na^D C_{DA} \rt) + \frac{1}{4} (\na_B\na^A Y^B) \na^D C_{DA}\\
&= \frac{1}{4} \na_B \lt( \na^B Y^A \na^D C_{DA}\rt) - \frac{1}{4} \na_B \lt( \na^A Y^B \na^D C_{DA} \rt) + \frac{1}{2} \na^D \lt( \na^A Y^1 C_{DA} \rt) + \frac{1}{2}\na^D \lt( Y^A C_{DA} \rt)\\
&= \frac{1}{2}\na_B \lt(\na^BY^A\na^D C_{DA}\rt) - \frac{1}{2} \na_A \lt( Y^1 \na_B C^{BA} \rt) + \frac{1}{2} \na^D \lt( \na^A Y^1 C_{DA} \rt) + \frac{1}{2}\na^D \lt( Y^A C_{DA} \rt)
\end{align*}
Finally  \eqref{RHS_line2} is equal to
\begin{align*}
\na_A \lt(Y^A m\rt)-\frac{1}{4}uY^1 |N|^2 + \frac{1}{4} Y^1 C_{BD}N^{BD} + \na_B \lt( \frac{u}{2} Y^1 \na_D N^{BD} \rt) - \na_A \lt( \frac{u}{2} \na_B Y^1 N^{AB} \rt). 
\end{align*}
\end{proof}

We rewrite \eqref{evolution} in terms of differential forms. Let $\slashed\epsilon = \epsilon_{AB} dx^A dx^B$ denote the area form of $(S^2,\sigma)$ and $\slashed d$ denote the differential of $S^2$.

Define the 2-form $\tilde{\omega}(u,x)$ on $\mathscr{I}^+$
\begin{align*}
\tilde{\omega} = \lt( Y^A ( N_A - \frac{1}{4}C_{AB}\na_D C^{BD} ) + \frac{1}{8}Y^1 C_{BD}C^{BD} + uY^1 \cdot 2m \rt)\slashed\epsilon - (\tilde{V} \lrcorner \slashed\epsilon) \wedge du.
\end{align*} 
Applying the identity
\begin{align*}
\slashed d \lt( \tilde{V} \lrcorner \slashed\epsilon \rt) = (\na_A \tilde{V}^A) \slashed\epsilon
\end{align*}
to time-dependent vector fields on $S^2$, we see that \eqref{evolution} is equivalent to
\begin{align*}
d\tilde{\omega} = \lt( -\frac{1}{4} N^{AB} ( \mathcal{L}_Y C_{AB} +uY^1 N_{AB} - Y^1 C_{AB} )\rt)\slashed\epsilon \wedge du.
\end{align*}
By Stokes' Theorem, we see that
\begin{align*}
Q_1 (u = u_0 + f)= \int_{ u=f(x)} \tilde{\omega}.
\end{align*}
On the section $u = f(x)$, we have
\begin{align*}
\lt( \tilde{V} \lrcorner \slashed{\epsilon} \rt) \wedge du  = \sum_{A < B} \lt( \tilde{V}^D \epsilon_{DA} \na_B f - \tilde{V}^D \epsilon_{DB} \na_A f \rt) dx^A \wedge dx^B = - \tilde{V}^D \na_D f \slashed{\epsilon}
\end{align*} and hence 
\begin{align*}
\int_{u=f(x)} \tilde{\omega} = \int_{S^2} Y^A \lt( N_A - \frac{1}{4}C_{AB}\na_D C^{BD} + \frac{1}{8}Y^1 C_{BD}C^{BD} + (u_0 + f)Y^1 \cdot 2m \rt) + \tilde{V}^D \na_D f
\end{align*}
where $N_A, C_{AB}, m, \tilde{V}^D$ are evaluated at $(u_0 + f(x),x)$ and
\begin{align*}
\int_{S^2} \tilde{V}^D \na_D f &= \int_{S^2} Y^A m \na_A f -\frac{1}{4}Y^A P_{BA}\na^B f + \frac{1}{2} \na_B Y^A \na^D C_{DA}\na^B f\\
&\quad + \int_{S^2} Y^A \lt( \frac{1}{2}C_{AB}N^{BD}\na_D f + \frac{1}{2}C_{AB}\na^B f \rt)\\
&\quad + \int_{S^2} -\frac{1}{2}Y^1 \na_B C^{AB}\na_A f + \frac{1}{2}\na^A Y^1 C_{AB}\na^B f\\
&\quad + \int_{S^2} \frac{1}{2}(u_0 + f)Y^1\na_D N^{BD} \na_B f - \frac{1}{2} (u_0 + f)\na_B Y^1 N^{AB}\na_A f.
\end{align*}

We turn to $Q_2(u=u_0 + f(x))$.
\begin{lem}\label{using transformation formula}
\begin{align*}
\int_{S^2} Y^A \lt( N_A' - 2m' \na_A f - \frac{1}{4}C'_{AB}\na_D C^{'DB} \rt) &= \int_{S^2} Y^A \lt( N_A + m \na_A f - \frac{1}{4}C_{AB}\na_D C^{DB} \rt)\\
&\quad + \int_{S^2} -\frac{1}{4}Y^A P_{BA}\na^B f + \frac{1}{2} \na_B Y^A \na^D C_{DA}\na^B f\\
&\quad + \int_{S^2} Y^A \lt( \frac{1}{2}C_{AB}N^{BD}\na_D f + \frac{1}{2}C_{AB}\na^B f \rt)\\
&\quad- \int_{S^2} Y^1 \na_B C^{AB}\na_A f.
\end{align*}
On the right-hand side the Bondi data $N_A, m, C_{AB}, N_{AB}, P_{BA}$ are evaluated at $(f(x),x)$ and then integrated over $S^2$.
\end{lem}
\begin{proof}
By the transformation formula of angular momentum and mass aspects, we obtain
\begin{align*}
N_A' - 2m' \na_A f &= N_A + m \na_A f - \frac{1}{4} P_{BA}\na^B f \\
&\quad - \frac{1}{2}\na_B \lt( \na^D C_{DA}(f)\na^B f \rt) + \underline{ \frac{1}{2}\na^D C_{DA}\Delta f }\\
&\quad + \frac{1}{2} \na_A \lt( \na^D C_{DB}(f) \na^B f \rt) -\frac{1}{2}\na^D C_{DB}\na_A\na^B f\\
&\quad -\frac{1}{4}\na^B \lt( N_{AB}(f)|\na f|^2 \rt) + \frac{1}{2}N_{AB}\na^B\na_D f \na^D f\\
&\quad + \frac{3}{4} N^B_D C_{AB}\na^D f - \frac{1}{2}N_{BD}\na^B\na^D f \na_A f
\end{align*}
where we rearranged
\begin{align*}
-\frac{3}{4}P_{BA} &= -\frac{1}{4} P_{BA}  -\frac{1}{2} \na_B\lt ( \na^D C_{DA}(f) \rt) + \frac{1}{2}\na^D N_{DA} \na_B f + \frac{1}{2}\na_A \lt( \na^D C_{DB}(f) \rt) - \frac{1}{2}\na^D N_{DB}\na_A f. \end{align*}

On the other hand,
\begin{align*}
-\frac{1}{4}C'_{AB}\na_D C^{'DB} &= -\frac{1}{4}C_{AB}\na_D C^{BD} - \frac{1}{4}C_{AB}N^{BD}\na_D f + \frac{1}{4}C_{AB}\na^B (\Delta + 2 )f\\
&\quad + \frac{1}{4}\na_D C^{BD}(2\na_A\na_B f \underline{- \Delta f \sigma_{AB} }) + \frac{1}{4}N^{BD}(2\na_A\na_B f - \Delta f \sigma_{AB} )\na_D f\\
&\quad -\frac{1}{4} (2\na_A\na_B f - \Delta f \sigma_{AB}) \na^B(\Delta + 2)f.
\end{align*}
Note that the last line vanishes when integrating against $Y^A$.
Per $\na^D C_{AD} = \na^D (C_{AD}(f)) - N_{AD}\na^D f$, the underlined term can be rearranged as
\begin{align*}
\na^D C_{AD} \Delta f = \na^D \lt( C_{AD}(f) \Delta f \rt) - C_{AD} \na^D\Delta f - N_{AD}\na^D f \Delta f.
\end{align*}
Putting these together, we get
\begin{align*}
&\int_{S^2} Y^A \lt( N_A' - 2m' \na_A f - \frac{1}{4}C'_{AB}\na_D C^{'DB} \rt) \\
&= \int_{S^2} Y^A \lt( N_A + m \na_A f - \frac{1}{4}C_{AB}\na_D C^{DB} \rt) + \int_{S^2} -\frac{1}{4}Y^A P_{BA}\na^B f + \frac{1}{2} \na_B Y^A \na^D C_{DA}\na^B f\\
&\quad - \int_{S^2} Y^1 \na^D C_{DB}\na^B f + \int_{S^2} Y^A \lt( \frac{1}{2}C_{AB}N^{BD}\na_D f + \frac{1}{2}C_{AB}\na^B f \rt) \\
&\quad + \frac{1}{2}\int_{S^2} Y^A \lt( N_{AB} F^B_D + N_D^B F_{BA} - N_{BE} F^{BE} \sigma_{AD} \rt) \na^D f.
\end{align*}
Here $F_{AB} = \na_A \na_B f - \frac{1}{2}\Delta f \sigma_{AB}$ and hence the integrand in the last integral vanishes by virtue of \eqref{product of symmetric traceless 2-tensors}.
\end{proof}
It remains to show that 
\begin{lem}
\begin{align*}
&\int_{S^2} \frac{1}{8}Y^1 C'_{BD} C^{'BD} + (u_0 + f)Y^1 \cdot 2m' \\
&= \int_{S^2} \frac{1}{8}Y^1 C_{BD}C^{BD} + (u_0 + f)Y^1 \cdot 2m \\
&\quad + \int_{S^2} \frac{1}{2}\na^A Y^1 C_{AB}\na^B f + \frac{1}{2}(u_0 + f)Y^1\na_D N^{BD} \na_B f - \frac{1}{2} (u_0 + f)\na_B Y^1 N^{AB}\na_A f\\
&\quad + \int_{S^2} \frac{1}{2} Y^1 \na_B C^{AB}\na_A f 
\end{align*}
\end{lem}
\begin{proof}
In the proof we omit the $(u_0 + f)$ dependence of Bondi data $C_{AB}, m$. First of all, we have
\[ C'_{BD}C^{'BD} = C_{BD}C^{BD} - 4 C_{AB}\na^A\na^B f + (-2\na_A\na_B f + \Delta f \sigma_{AB})(-2\na^A\na^B f + \Delta f \sigma^{AB}) \]
and note that the last term vanishes when integrating against $Y^1$. Next, we rewrite the transformation of the mass aspect as
\[ 2m' = 2m + \frac{1}{2}\na^B N_{AB} \na^A f + \frac{1}{2} \na^B (N_{AB}\na^A f) \]
to get
\begin{align*}
\int_{S^2} (u_0 + f)Y^1 \cdot 2m'  &= \int_{S^2} (u_0 + f)Y^1 \cdot \lt( 2m + \frac{1}{2}\na_D N^{BD} \na_B f \rt) - \frac{1}{2}(u_0 + f)\na_B Y^1 N^{AB}\na_A f \\
&\quad - \int_{S^2} \frac{1}{2} Y^1 N^{AB} \na_A f \na_B f.
\end{align*}
Finally, we have
\begin{align*}
\int_{S^2} -\frac{1}{2}Y^1 C_{AB}(f) \na_A\na_B f = \int_{S^2} \frac{1}{2} \na^B Y^1 C_{AB}\na^B f + \frac{1}{2} Y^1 \lt( \na^B C_{AB} + N_{AB}\na^B f\rt)\na^A f.
\end{align*}
\end{proof}

\appendix

\section{Identities of symmetric traceless 2-tensors}
\begin{lem}\label{symmetric_traceless}
Let $\sigma$ be a metric of Gauss curvature 1 on $S^2$. Let $T$ be a symmetric traceless 2-tensor. Then
\begin{align*}
\na_a T_{bc} + \na_b T_{ac} = 2\na_c T_{ab} + 2\na^d T_{cd} \sigma_{ab} - \na^d T_{ad} \sigma_{bc} - \na^d T_{bd} \sigma_{ac}.
\end{align*} 
\end{lem}
\begin{proof}
It is well-known that $T$ admits the decomposition
\[ T_{ab} = \na_a\na_b h - \frac{1}{2}(\Delta h) \sigma_{ab} + \frac{1}{2} (\epsilon_{ad} \na_b \na^d f + \epsilon_{bd} \na_a \na^d f )\]
for functions $h$ and $f$. See Theorem B.2 of \cite{KWY} for a proof. 

If $T = \na_a\na_b h - \frac{1}{2}(\Delta h) \sigma_{ab}$, then the assertion follows from Ricci identity. If $T_{ab} = \frac{1}{2} (\epsilon_{ad} \na_b \na^d f + \epsilon_{bd} \na_a \na^d f )$, then Proposition A.1 of \cite{KWY} implies
\begin{align*}
\na_a T_{bc} + \na_b T_{ac} - 2 \na_c T_{ab} = -\frac{1}{2} \epsilon_{ac} \na_b (\Delta+2)f - \frac{1}{2}\epsilon_{bc} \na_a (\Delta+2)f.
\end{align*}
On the other hand, from the relation $\sigma_{ab} = \frac{1}{2} \lt( \epsilon_{ae}\epsilon_b^{\;\;e} + \epsilon_{be}\epsilon_a^{\;\;e} \rt)$ and $\epsilon_{cd}\epsilon_{ae} = \sigma_{ca}\sigma_{de}-\sigma_{ce}\sigma_{da},$ we get
\begin{align*}
&2\na^d T_{cd} \sigma_{ab} - \na^d T_{ad} \sigma_{bc} - \na^d T_{bd} \sigma_{ac} \\
&= \frac{1}{2}\epsilon_{cd}\na^d(\Delta+2)f \lt( \epsilon_{ae}\epsilon_b^{\;\;e} + \epsilon_{be}\epsilon_a^{\;\;e}\rt) - \frac{1}{2}\epsilon_{ad}\na^d(\Delta+2)f \sigma_{bc} - \frac{1}{2}\epsilon_{bd}\na^d(\Delta+2)f \sigma_{ac}\\
&= -\frac{1}{2} \epsilon_{ac} \na_b (\Delta+2)f - \frac{1}{2}\epsilon_{bc} \na_a (\Delta+2)f.
\end{align*}
This completes the proof.
\end{proof}

The following is a consequence of the fact 
\begin{align}
\label{product of symmetric traceless 2-tensors} T_{ad} S^d_b + T_{bd}S^d_a = T_{de}S^{de} \sigma_{ab}
\end{align} for any symmetric traceless 2-tensors $T,S$ on a surface.
\begin{lem}\label{multiple of metric}
Let $T$ be a symmetric traceless 2-tensor and $S$ be a symmetric 2-tensor. Then
\begin{align*}
(\tr S) T_{ab} - T_{ad}S^d_b - T_{bd}S^d_a =- T_{de}S^{de} \sigma_{ab}.
\end{align*}
\end{lem}

  \section{The general scheme to compute metric coefficients}

The spacetime metric takes the form
\[\bar{g}_{\bar{\alpha}\bar{\beta}} d\bar{x}^{\bar{\alpha}}d\bar{x}^{\bar{\beta}},\] where $\bar{\alpha}, \bar{\beta}=\bar{u}, \bar{r}, \bar{x}^A, A=2, 3$.

In terms of $U, V, h_{AB}, W^A$ in \eqref{Bondi-Sachs2}, we have 
\[\begin{split}\bar{g}_{\bar{u}\bar{u}}&=-{U}(\bar{u}, \bar{r}, \bar{x}){V}(\bar{u}, \bar{r}, \bar{x})+\bar{r}^2 h_{AB}(\bar{u}, \bar{r}, \bar{x}) W^A (\bar{u}, \bar{r}, \bar{x})W^B(\bar{u}, \bar{r}, \bar{x})\\
\bar{g}_{\bar{u}\bar{r}}&=-{U}(\bar{u}, \bar{r}, \bar{x})\\
\bar{g}_{\bar{u}{A}}&=\bar{r}^2 h_{AB} (\bar{u}, \bar{r}, \bar{x})W^B(\bar{u}, \bar{r}, \bar{x})\\
\bar{g}_{{A}{B}}&=\bar{r}^2 h_{AB} (\bar{u}, \bar{r}, \bar{x}).\end{split}\] 
The expansions of $U,V, h_{AB},W^A$ imply
\[\begin{split}\bar{g}_{\bar{u}\bar{u}}&=\bar{g}_{\bar{u}\bar{u}}^{(k)}(\bar{u}, \bar{x})\bar{r}^k\\
\bar{g}_{\bar{u}\bar{r}}&=\bar{g}_{\bar{u}\bar{r}}^{(k)}(\bar{u}, \bar{x})\bar{r}^k\\
\bar{g}_{\bar{u}{A}}&=\bar{g}_{\bar{u}{A}}^{(k)}(\bar{u}, \bar{x})\bar{r}^k\\
\bar{g}_{{A}{B}}&=\bar{g}_{{A}{B}}^{(k)}(\bar{u}, \bar{x})\bar{r}^k\end{split}\]

The BMS transformations can be written as expansions in ${r}^k$:

\[\begin{split}\bar{u}&=\bar{u}^{(k)}(u, x) r^k, k=0, -1, -2, \cdots\\
\bar{r}&=\bar{r}^{(k)}(u, x) r^k, k=1, 0, -1, -2, \cdots\\ 
\bar{x}^{{A}}&=\bar{x}^{A(k)}(u, x) r^k, k=0, -1, -2, \cdots\end{split}\]

We claim that the above expansions can be used to rewrite 
\[\bar{g}_{\bar{\alpha}\bar{\beta}}^{(k)}(\bar{u}, \bar{x})\bar{r}^k=
{g}_{\bar{\alpha}\bar{\beta}}^{(k)}({u}, {x}){r}^k,\] where the right hand side is completely in terms of the coordinate variables $(u, r, x)$, and in particular as an expansion of $r^k$. This is done by way of the Taylor expansions near $r=\infty$. For example, for a function $h(\bar{u}, \bar{x})$, we obtain: 
$h(\bar{u}, \bar{x})=h^{(0)}(u, x)+h^{(-1)}(u, x) r^{-1}+\cdots$,
where
\[h^{(0)}(u, x)=h(\bar{u}^{(0)}(u,x), \bar{x}^{(0)}(u, x))\] and 
\[h^{(-1)}(u, x)=\frac{\partial h}{\partial \bar{u}}(\bar{u}^{(0)}(u,x), \bar{x}^{(0)}(u, x))\bar{u}^{(-1)}(u, x)
+\frac{\partial h}{\partial \bar{x}^{A}}(\bar{u}^{(0)}(u,x), \bar{x}^{(0)}(u, x))\bar{x}^{{A} (-1)}(u, x).
\]

This is how the first expansion formula of $\bar{r}^2 h_{AB}$ in Section 4.1 is obtained.

\section{Appendix C.5 of \cite{CJK}}
Appendix C.5 of \cite{CJK} gives the transformation formula of mass and angular momentum aspects under supertranslations. The approach is through covariant (inverse) metrics and is more natural. However, several terms were overlooked in the formula for the angular momentum aspect. We go over their argument, add details, and show that their transformation formula of angular momentum aspect after corrections coincides with our Theorem \ref{angular momentum aspect}. For the reader's convenience, the equation numbers adhere to those in \cite{CJK}.

Let $x = \frac{1}{r}$. The asymptotics of the covariant (inverse) metric are given by
\begin{align}
g^{xx} = g(dx, dx) = x^4 (1 - 2M x) + O(x^6), \tag{C.102}\\
g^{xu} = g(dx,du) = x^2 \lt( 1 + \frac{1}{16} x^2 \chi^{AB}\chi_{AB} \rt) + O(x^5), \tag{C.103}\\
g^{xA} = g(dx,dx^A) = x^4 \lt[ -\frac{1}{2} \chi^{AB}_{\;\;\;\;\;\;\|B} + x(2N^A_{(CJK)} + B^A) \rt] + O(x^6), \tag{C.104}\\
g^{BA} = g(dx^B, dx^A) = x^2 \lt[  \breve{h}^{AB} - x \chi^{AB} + \frac{1}{2}x^2 \chi^{AC}\chi^B_{\;\;C} \rt] + O(x^5), \tag{C.105}
\end{align}
where
\begin{align}
B^A := \frac{1}{2} \chi^A_{\;\;B} \chi^{BC}_{\;\;\;\;\;\;\|C} + \frac{1}{16}(\chi^{CD}\chi_{CD})^{\|A}. \tag{C.106}
\end{align} 
Here $\breve{h}^{AB}$ is the inverse of the unit metric on $S^2$ and we raise and lower indices with respect to it; $\|A$ denotes the covariant derivative with respect to $\breve{h}$.

We consider coordinate transformations of the form
\begin{align}
\bar u = u - \lambda + x u_1 + x^2 u_2 + x^3 u_3 + \cdots, \tag{C.107}\\
\bar x = x + x^2 \rho_2 + x^3 \rho_3 + \cdots, \tag{C.108}\\
\bar x^A = x^A + x x_1^A + x^2 x_2^A + x^3 x_3^A + \cdots, \tag{C.109}
\end{align}
which preserve the Bondi-Sachs coordinates. 

One finds the formula of $(x_1^A, u_1), (x_2^A, \rho_2), u_2, \rho_3$ in (C.112), (C.113), (C.114), (C.115) of \cite{CJK} as well as the transformation rules of the shear tensor and the mass aspect
\begin{align}
\bar\chi_{AB} = \chi_{AB} - 2 \lambda_{\|AB} + \breve{h}_{AB} \Delta \lambda, \tag{C.117}\\
\bar M = M + \frac{1}{2} \chi^{AB}_{\;\;\;\;\;\;,u\|B}\lambda_{\|A} + \frac{1}{4} \chi^{AB}_{\;\;\;\;\;\;,u} \lambda_{\|AB} + \frac{1}{4} \chi^{AB}_{\;\;\;\;\;\;,uu}\lambda_{\|A}\lambda_{\|B} \tag{C.119}
\end{align}
where $\bar\chi_{AB}$ and $\bar M$  are evaluated at $(u-\lambda, x^A)$ and the Bondi functions on the right-hand side are evaluated at $(u,x^A)$.

To see the coefficients of coordinate transformations coincide with ours, say $g^{(-2)A} = x_2^A$, set $\lambda = -f$ and use (C.117) with $\bar\chi_{AB} = C_{AB}$. On the other hand, since the supertranslation (C.107) is the inverse of \eqref{bar u}, to verify the consistency of the transformation formula of Bondi functions, one simply compare (C.117) and (C.119) with \eqref{shear} and Theorem \ref{mass aspect} assuming $\lambda = f$ to see that the transformation formula of Bondi functions coincide.

To get the transformation of angular aspect, we need the time derivative $\dot x_3^A$. By (C.112) of \cite{CJK}, we have
\begin{align*}
d\bar u &= (1 + x^2 \dot u_2)du + u_1 dx + (-\lambda_{\|B} + x \pl_B u_1 + x^2 \pl_B u_2)dx^B + \cdots,\\
d\bar x^A &= \lt( \delta^A_C + x \pl_C x_1^A + x^2 \pl_C x_2^A \rt)dx^C + \lt( x_1^A + 2x x_2^A + 3x^2 x_3^A \rt)dx + x^2 \dot x_2^A du + \cdots
\end{align*}
and the $x^4$ term of $g(d\bar u, d\bar x^A)=0$ yields
\begin{align*}
0 &= x_1^A \cdot \frac{1}{16} \chi^{BC}\chi_{BC} + 3 x_3^A + x_1^A \dot u_2\\
&\quad + u_1 \lt( -\frac{1}{2}\chi^{AB}_{\;\;\;\;\;\;\|B} + x_1^A + \dot x_2^A \rt)\\
&\quad - \lambda_{\|B} \lt( \frac{1}{2} \chi^{AC}\chi^B_{\;\;C} - \pl_C x_1^A \chi^{BC} + \breve{h}^{BC} \pl_C x_2^A \rt)\\
&\quad + \pl_B u_1 \lt( -\chi^{AB} + \breve{h}^{BC}\pl_C x_1^A \rt)\\
&\quad + \pl_B u_2 \breve{h}^{AB} + \frac{1}{2} \lambda^{\|A} x_1^B \chi_{BC}^{\;\;\;\;\;\;\|C}.
\end{align*}     
Differentiating in $u$ and using (C.112-C.114) of \cite{CJK}, we obtain
\begin{align}\tag{Corrected form of (C.116)}
\begin{split}
\dot{x}_3^A &= -\frac{1}{12} \ddot{\chi}^{CB} \lambda_{\|C} \lambda_{\|B} \lambda^{\|A} - \frac{1}{12} \ddot{\chi}^{CA} \lambda_{\|C} \lambda_{\|B}\lambda^{\|B} - \frac{1}{12} \lt( \dot{\chi}^{CB} \lambda_C \lambda_B\rt)^{\|A}\\
&\quad -\frac{1}{24}\chi_{CB}\dot{\chi}^{CB} \lambda^{\|A} - \frac{1}{6}\lambda_{\|C} \lt( \lambda_{\|B} \dot\chi^{BA}\rt)^{\|C} + \frac{1}{6}\lambda^{\|C} \lt( \dot\chi_{CB}\chi^{BA} + \chi_{CB}\dot\chi^{BA}\rt)\\
&\quad -\frac{1}{3} \lambda^{\|C} \pl_B\lt( \lambda^{\|A}\rt)\dot\chi_C^{\;\;B} - \frac{1}{6} \lt( \lambda_{\|C}\lambda^{\|C} \rt)_{\|B} \dot\chi^{BA} - \frac{1}{6}\lambda_{\|C} \lambda^{\|A} \dot\chi^{BC}_{\;\;\;\;\;\;\|B}\\
&\quad -\frac{1}{12} \lambda_{\|C}\lambda^{\|C} \dot\chi^{BA}_{\;\;\;\;\;\;\|B} ,
\end{split}
\end{align}
Note that the fourth term on the right-hand side of the corrected form above is different from the fourth term on the right-hand side of  the original form of (C.116).

The transformation of angular aspect is obtained from $g(d\bar x, d\bar x^A)$. From
\[ d\bar x = (1 + 2 \rho_2)dx + (x^2 \pl_B \rho_2 + x^3\pl_B \rho_3) dx^B + x^3 \dot\rho_3 du + \cdots, \]
we get, by (C.102-C.105),
\begin{align*}
g(d\bar x, d\bar x^A) &= x^4 \lt( -\frac{1}{2} \chi^{AB}_{\;\;\;\;\;\;\|B} + \dot x_2^A + x_1^A + \rho_2^{\;\;\|A} \rt)\\
&\quad + x^5 \Big[ -\rho_2 \chi^{AB}_{\;\;\;\;\;\;\|B} + \pl_c x_1^A \lt( -\frac{1}{2}\chi^{CB}_{\;\;\;\;\;\;\|B} + 2N^A + B^A \rt) \\
&\qquad\qquad +2\rho_2 \dot x_2^A + \dot x_3^A \\
&\qquad\qquad + 2\rho_2 x_1^A + 2x_2^A - 2Mx_1^A \\
&\qquad\qquad + \rho_3^{\;\;A} + \rho_2^{\;\;B}\pl_B x_1^A - \pl_B \rho_2 \chi^{BA} + \dot\rho_3 x_1^A \Big] +\cdots.
\end{align*}
On the other hand, we expand  $g(d\bar x, d \bar x^A) = \bar x^4 ( -\frac{1}{2} \bar\na_B \bar\chi^{AB} ) + \bar x^5(2\bar N^A + \bar B^A)$ (the right-hand-side is evaluated at $(\bar u, \bar x^A)$) by (C.107-C.109) to get 
\begin{align*}
g(d\bar x, d\bar x^A) &= x^4 (-\frac{1}{2}\bar\na_B \bar\chi^{AB}) \\
&\quad + x^5 \lt[ 4\rho_2 (-\frac{1}{2}\bar\na_B\bar\chi^{AB}) + u_1 \pl_{\bar u}(-\frac{1}{2}\bar\na_B\bar\chi^{AB}) + x_1^C \frac{\pl}{\pl\bar x^C} \rt] + \cdots.
\end{align*} Consequently,
\begin{align}
\bar\na_B \bar\chi^{AB} = \chi^{AB}_{\;\;\;\;\;\;\|B} - 2 \dot x_2^A - 2 \rho_2^{\;\;\|A} - 2x_1^A \tag{C.118}
\end{align}
and
\begin{align}\tag{C.120}
\begin{split}
2\bar N^A + \bar B^A &= 2N^A + B^A - 2Mx_1^A - 4\rho_2 \rho_2{\;\;\|A} + \rho_2 \chi^{AB}_{\;\;\;\;\;\;\|B} -2\rho_2 x_1^A - 2\rho_2 \dot x_2^A \\
&\quad + 2x_2^A - \frac{1}{2}\chi^{CB}_{\;\;\;\;\;\;\|C} x_{1}^A\mbox{}_{,B} + \dot\rho_3 x_1^A + \rho_3^{\;\;\|A} + \dot x_3^A - \rho_{2,B}\chi^{AB} + \rho_2^{\;\;,B} x_1^A\mbox{}_{,B}\\
&\quad + \frac{1}{2} \lt[ \lt( u_1 + x_1^C \lambda_{\|C} \rt)\pl_u + x_1^C \pl_C \rt]\lt( \chi^{AB}_{\;\;\;\;\;\;\|B} - 2 \dot x_2^A - 2 \rho_2^{\;\;\|A} - 2x_1^A \rt) 
\end{split}
\end{align}
where the left-hand-sides are evaluated at $(u-\lambda, x^A)$. Note that in deducing (C.120) we have used 
\begin{align*}
\pl_{\bar u} \bar\na_B \bar\chi^{AB} = \pl_u (\chi^{AB}_{\;\;\;\;\;\;\|B} - 2 \dot x_2^A - 2 \rho_2^{\;\;\|A} - 2x_1^A),\\
\pl_{\bar u} \bar\na_B \bar\chi^{AB} (-\lambda_{\|C}) + \frac{\pl}{\pl \bar x^C} \bar\na_B \bar\chi^{AB} = \pl_C (\chi^{AB}_{\;\;\;\;\;\;\|B} - 2 \dot x_2^A - 2 \rho_2^{\;\;\|A} - 2x_1^A),
\end{align*}
which, obtained by differentiating (C.118), are equivalent to (C.121) of \cite{CJK}.
  
By (C.122) of \cite{CJK}, we obtain the corrected transformation formula for $N^A_{(CJK)}$:
\begin{align}\tag{Corrected form of (C.123)}
\begin{split}
2 \bar N^A_{(CJK)} &= 2 N^A_{(CJK)} - 2 M \lambda^{\|A} + \frac{1}{2} \lambda^{\|C} \chi^{AB}
_{\;\;\;\;\;\;\|BC} - \frac{1}{2}\lambda^{\|C} \chi_{CB}^{\;\;\;\;\;\;\|BA}\\
&\quad + \frac{1}{6}\lambda^{\|C}\dot\chi^{AB}\chi_{BC} - \frac{1}{12}\lt( \dot\chi^{BC} \lambda_{\|B} \lambda_{\|C}\rt)^{\|A}\\
&\quad + \frac{1}{3}\lambda^{\|C} \dot\chi^{AB}_{\;\;\;\;\;\;\|C} \lambda_{\|B} + \frac{2}{3} \lambda^{\|C} \dot\chi_{CB} \lambda^{\|AB} - \frac{2}{3} \lambda^{\|A} \lambda_{\|C} \dot\chi^{BC}_{\;\;\;\;\;\;\|B}\\
&\quad + \frac{1}{6} \lambda^{\|C} \lambda_{\|C} \dot \chi^{AB}_{\;\;\;\;\;\;\|B} - \frac{1}{6} \lambda^{\|A} \dot\chi^{CD}\chi_{CD} - \frac{1}{3} \lambda_{\|C} \dot\chi^{BC} \chi^A_{\;\;B} \\
&\quad + \frac{1}{6} \lambda^{\|C} \lambda_{\|C} \ddot\chi^{AB} \lambda_{\|B} - \frac{1}{3} \lambda^{\|A} \lambda_{\|C} \ddot\chi^{BC} \lambda_{\|B}\\
&\quad -\frac{1}{4} \lt( \dot\chi^{BC} \lambda_{\|B} \lambda_{\|C}\rt)^{\|A}, 
\end{split}
\end{align}
Note that there are differences in  the first, second, and fifth line between the above corrected form of (C.123) and the original form of (C.123); the last term (in the last line), coming from $\rho_3^{\;\;\|A}$, was missing in the original form of (C.123). 

Using the identity $\dot\chi^{AB}\chi_{BC} + \dot\chi_{CB} \chi^{BA} = \dot\chi^{BD}\chi_{BD} \delta^A_C$ to rewrite the term  $\frac{1}{6}\lambda^{\|C}\dot\chi^{AB}\chi_{BC}$ and Lemma \ref{symmetric_traceless} to rewrite the term $\frac{1}{3}\lambda^{\|C} \dot\chi^{AB}_{\;\;\;\;\;\;\|C} \lambda_{\|B}$ in the above corrected form of (C.123), the formula becomes
\begin{align*}
2 \bar N^A_{(CJK)} &= 2 N^A_{(CJK)} - 2 M \lambda^{\|A} + \frac{1}{2} \lambda^{\|C} \chi^{AB}
_{\;\;\;\;\;\;\|BC} - \frac{1}{2}\lambda^{\|C} \chi_{CB}^{\;\;\;\;\;\;\|BA}\\
&\quad -\frac{1}{2} \dot\chi^{BC} \chi^A_{\;\;B} \lambda_{\|C} - \cancel{ \frac{1}{3} \lt( \dot\chi^{BC} \lambda_{\|B} \lambda_{\|C}\rt)^{\|A}  }\\
&\quad + \frac{1}{3} \lt( \cancel{ \dot\chi^{BC\|A} \lambda_{\|B}\lambda_{\|C}  } + \dot\chi^{AB}_{\;\;\;\;\;\;\|B} \lambda_{\|C} \lambda^{\|C} - \dot\chi^{BC}_{\;\;\;\;\;\;\|B} \lambda_{\|C}\lambda^{\|A} \rt) + \cancel{ \frac{2}{3} \lambda^{\|C} \dot\chi_{CB} \lambda^{\|AB} }  - \frac{2}{3} \lambda^{\|A} \lambda_{\|C} \dot\chi^{BC}_{\;\;\;\;\;\;\|B}\\
&\quad + \frac{1}{6} \lambda^{\|C} \lambda_{\|C} \dot \chi^{AB}_{\;\;\;\;\;\;\|B} + \frac{1}{6} \lambda^{\|C} \lambda_{\|C} \ddot\chi^{AB} \lambda_{\|B} - \frac{1}{3} \lambda^{\|A} \lambda_{\|C} \ddot\chi^{BC} \lambda_{\|B},  
\end{align*} where the three cross-out terms are cancelled. 
Comparing (C.104) in \cite{CJK} and equation \eqref{defn N^A}, we have $
N^A = -3 N^A_{(CJK)}$. In equation \eqref{final_angular_formula}, setting $K\equiv1$ and plugging the relations between $N^A$ and $N^{A}_{(CJK)}$, we obtain
\begin{align*}
2 \bar N^A_{(CJK)} &= 2 N^A_{(CJK)} - 2 M \lambda^{\|A} + \frac{1}{2} \lambda^{\|C} \chi^{AB}
_{\;\;\;\;\;\;\|BC} - \frac{1}{2}\lambda^{\|C} \chi_{CB}^{\;\;\;\;\;\;\|BA} \\
&\quad + \frac{1}{2}\dot\chi^{AB}_{\;\;\;\;\;\;\|B} \lambda_{\|C}\lambda^{\|C} - \frac{1}{2}\dot\chi^{BC} \chi^A_{\;\;B} \lambda_{\|C} - \frac{1}{3} \ddot\chi^{BC} \lambda_{\|B} \lambda_{\|C} \lambda^{\|A}\\
&\quad - \dot\chi^{BC}_{\;\;\;\;\;\;\|C} \lambda_{\|B} \lambda^{\|A} + \frac{1}{6} \ddot\chi^A_{\;\;B} \lambda^{\|B} \lambda_{\|C} \lambda^{\|C}.
\end{align*} 
Therefore, the corrected form of (C.123) coincides with our equation \eqref{final_angular_formula}.

\thanks{The authors would like to thank Jordan Keller for discussions at early stage of this work. P.-N. Chen is supported by NSF grant DMS-1308164 and Simons Foundation collaboration grant \#584785, M.-T. Wang is supported by NSF grant DMS-1810856 and DMS-2104212, Y.-K. Wang is supported by MOST Taiwan grant 107-2115-M-006-001-MY2, 109-2628-M-006-001-MY3. The authors would like to thank the National Center for Theoretical Sciences at National Taiwan University where part of this research was carried out. This material is based upon work supported by the National Science Foundation under Grant No. DMS-1810856 and  DMS-2104212 (Mu-Tao Wang).}

\end{document}